\documentclass[11pt,reqno]{amsart}
\usepackage{amsmath, amsthm, amssymb}
\usepackage{enumitem}
\usepackage{multicol}
\usepackage{bbm}
\usepackage{mathtools}

\usepackage{fullpage}
\usepackage[colorlinks=true, linkcolor=blue, citecolor=blue, urlcolor=blue]{hyperref}

\usepackage{amsmath}

\newtheorem{thm}{Theorem}[section]

\newtheorem{lem}[thm]{Lemma}
\newtheorem{cor}[thm]{Corollary}

\newcommand{\doi}[1]{\href{https://dx.doi.org/#1}{{\small DOI:#1}}}

\newcommand{\googlebooks}[1]{(preview at \href{https://books.google.com/books?id=#1}{google books})}

\newcommand{\numdam}[1]{}

\newtheorem{defn}[thm]{Definition}
\newtheorem{quest}[thm]{Problem}

\theoremstyle{remark}
\newtheorem{remark}[thm]{Remark}

\usepackage{graphicx,tikz}
\usetikzlibrary{decorations.pathreplacing}
\usetikzlibrary{arrows,backgrounds,patterns.meta}
\usetikzlibrary{positioning,shadings,cd}
\usetikzlibrary{calc}
\usepackage[numbers,sort&compress]{natbib}
\title{Universal fusion category symmetries on tensor products of infinite-dimensional Hilbert spaces}
\author{Ian Bunner and Corey Jones  }
\address{Department of Mathematics,
North Carolina State University, Raleigh, NC 27695, USA}
\date{}

\begin{document}

\begin{abstract}
We show that anyon chains, after stabilizing with infinite-dimensional ancilla spaces, factorize locally as tensor products of infinite-dimensional Hilbert spaces. This implies that any unitary fusion category can be realized as symmetries on a tensor product of infinite-dimensional Hilbert spaces. We then show that any two anyon chains with the same symmetry category are related by a symmetry-compatible locality-preserving unitary after stabilizing with infinite-dimensional ancilla, showing that for a fixed fusion category, there is a single stable equivalence class of symmetry realizations on the lattice via anyon chains. As a corollary of our proof, we show that the physical boundary algebras of Levin-Wen type models are bounded spread isomorphic after stabilization if and only if they have the same bulk topological order.
\end{abstract}

\maketitle
 
\tableofcontents

\section{Introduction}

Non-invertible symmetries have become an important tool for organizing the universal features of both quantum many-body systems and quantum field theories (see \cite{GKSW, PhysRevResearch.2.043086,SCHAFERNAMEKI20241,shao2024whatsundonetasilectures} for overview and references). In 1+1D, non-invertible symmetries have been studied on the lattice via topological defects, matrix product operators and weak Hopf algebras \cite{MR3543452,aasen2020topologicaldefectslatticedualities,PRXQuantum.4.020357,PRXQuantum.5.010338, In22, ahmad2026facesnoninvertiblesymmetries, MR3719546,TW24,TW242,bhardwaj2025latticemodelsphasestransitions, PhysRevB.111.054432, benini2026entanglementasymmetryhighernoninvertible}. These can all be viewed as manifestations of \textit{fusion category symmetry}. Recently, a general definition of fusion category symmetry in infinite volume was introduced \cite{evans2026operatoralgebraicapproachfusion}, inspired by the framework of SymTFT/topological holography \cite{ABGIH, PhysRevResearch.2.033417, PhysRevB.108.075105, 10.21468/SciPostPhysCore.6.4.066, MR3763324, PhysRevResearch.2.043044,KLWZZ,PhysRevResearch.1.033054,LZI,LZII, LZIII,chatterjee2024emergentgeneralizedsymmetrymaximal} and non-invertible symmetries in the theory of conformal nets \cite{MR3595480} (see also \cite{DongNgRenXu2025GeneralizedSymmetriesFusionActions}). This makes it possible to mathematically formalize and rigorously approach results from the physics literature model-independently, including the classification of categorical symmetry protected topological phases \cite{PhysRevLett.133.161601, GarreRubio2023classifyingphases, evans2026operatoralgebraicapproachfusion}, anomaly enforced gaplessness \cite{TW24,10.21468/SciPostPhys.16.6.154, evans2026operatoralgebraicapproachfusion}, and the structure of categorical dualities \cite{PRXQuantum.4.020357,PRXQuantum.5.010338,MR3543452,aasen2020topologicaldefectslatticedualities, jones2024dhr, JoLi24, jones2025quantumcellularautomatacategorical, JonesYang2026CategoricalDualityOperators}.

As part of this program, it is important to classify kinematical realizations of a symmetry category $\mathcal{C}$. We need to be able to identify when various local Hilbert space structures on the lattice host $\mathcal{C}$-symmetries, and in how many different ways. A classification of kinematical realizations of a symmetry category is distinct from and prior to the classification of symmetric states or Hamiltonians of a fixed realization of that symmetry on the lattice. Indeed, it is possible in principle that the solution to the classification of symmetric phases may depend on the specific kinematic realization in question.

The default kinematical structures used in quantum many-body systems are typically \textit{spin systems}, whose local Hilbert spaces are tensor products of finite-dimensional Hilbert spaces assigned to sites. However, it was recently shown that a fusion category is realizable as symmetries with this kind of local tensor product structure if and only if all of the quantum dimensions of the objects are integers \cite{evans2026operatoralgebraicapproachfusion}\footnote{Recent work \cite{LuChatterjeeTantivasadakarn2026HopfIsing, JonesYang2026CategoricalDualityOperators, Inamura2026RemarksNoninvertibleTensorProduct, WenInamuraSchaferNameki2026NonInvertible} using a generalized notion of non-invertible symmetry mixing with QCAs shows that one can extend this to weakly integral categories. }. This drastically restricts the class of fusion categories realizable on spin chains, excluding a majority of interesting examples. 

There is a standard class of alternative local kinematic structures that are commonly used to model fusion category symmetries, called \textit{anyon chains} \cite{Feiguin2007,MR3719546, PRXQuantum.4.020357, MR3614057, jones2025quantumcellularautomatacategorical}. The local Hilbert space assigned to an interval $I\subseteq \mathbbm{Z}$ in an anyon chain is generally \textit{not} the tensor product of local Hilbert spaces assigned to sites, but a nontrivial subspace of this tensor product defined by local constraints. Mathematically, anyon chain realizations of $\mathcal{C}$-symmetry are parameterized by particular choices of algebraic data associated to $\mathcal{C}$ (a module category and a generating object of the dual, see Section \ref{sec:anyonchains}), but in general there are \textit{infinitely many} of these choices. This suggests a zoo of local Hilbert space structures to model $\mathcal{C}$-symmetry, even for a fixed $\mathcal{C}$.

However, following the principle of universality, one expects most of the variations between distinct anyon chain realizations of the same symmetry category to be inessential, reflecting auxiliary details of a microscopic kinematical description rather than intrinsic features of the symmetry itself. This raises the question: is there a physically motivated notion of equivalence relation on microscopic realizations of fusion category symmetry such that all anyon chain realizations of the same category are equivalent? In the case of ordinary spin systems, where the problem of finding the right equivalence relation to reflect universal features has been considered extensively, the typical operational approach to defining equivalence (of states or Hamiltonians) in the context of gapped phases consists of two components \cite{Hastings2005Quasiadiabatic, chen2010local, chen2011classification, Schuch2011MPS}. We first allow the addition of \textit{ancilla} at every site, which addresses the possible differences between the local Hilbert space dimensions of the two models. Then one asks for (some kind of) locality-preserving unitary between the larger Hilbert spaces that preserves the structure in question. The addition of ancilla is physically reasonable for capturing universality, since real physical systems typically have access to larger local Hilbert spaces, which are ignored in a given model for mathematical expedience. 

It is natural to extend this kind of equivalence relation to symmetry realizations. Indeed, this perspective has been taken in the classification of locality-preserving actions of finite groups on spin chains \cite{seifnashri2025disentanglinganomalyfreesymmetriesquantum, BolsDeRoeckDeWildeOliveira2025Classification}. In this case, the above mentioned equivalence relation is augmented by allowing the ancilla to carry an on-site action of the group, and the new realization is the diagonal action of $G$ on the composite system. This approach yields satisfactory classification results. It is not obvious how to extend this definition to the case of fusion categories, since there is no analogue of adding ancilla with $\mathcal{C}$-symmetry for general fusion categories: tensoring two $\mathcal{C}$-symmetric chains yields a chain with $\mathcal{C}\boxtimes \mathcal{C}$ symmetry, and there is no generic analogue of the diagonal map used for groups. One could instead just add ancilla with no symmetry, and while this makes sense mathematically, it fails to yield a single equivalence class of kinematical realization (for example, see the discussion in \cite[Remark 10]{jones2025localtopologicalorderboundary}). 

In this paper, we propose that the most natural equivalence relation for fusion category symmetry should allow for the addition of an \textit{infinite-dimensional} space of ancilla at each site. We call this process \textit{stabilization}. We then define an equivalence relation we call \textit{stable equivalence} on the class of kinematical realizations of a fusion category $\mathcal{C}$ by asking for the existence of locality-preserving unitaries between stabilizations that are compatible with the symmetry (in a sense made precise below). We show that the apparent proliferation of anyon chain realizations is not a universal feature: after allowing infinite-dimensional ancilla, all anyon chains admit a local tensor product decomposition, and all realizations of a fixed fusion category on anyon chains are related by a symmetry-compatible locality-preserving unitary.

\bigskip

\textbf{Main results.}

\medskip

\begin{enumerate}
\item 
For any anyon chain, its stabilization admits a local decomposition into a tensor product of infinite-dimensional Hilbert spaces. 

\medskip

\item 
Given two anyon chains with the same symmetry fusion category, there is a locality-preserving unitary between their stabilizations that is compatible with the symmetry.

\end{enumerate}

\bigskip

This shows that the kinematical obstructions to realizing fusion category symmetries on spin chains are in some sense unstable, since they disappear after adding infinite-dimensional local ancilla. Furthermore, for a fixed fusion category $\mathcal{C}$, there is indeed a single universality class of anyon chain realizations. We reiterate an important point: by a single universality class, we are referring to universality of \textit{kinematical realization}, not universality of symmetric states or Hamiltonians. Indeed, stabilization remembers topological invariants of symmetric states. If we have a symmetric state $\phi$ on an anyon chain with $\mathcal{C}$-symmetry, there is an associated algebra object $L_{\phi}\in \mathcal{Z}(\mathcal{C})$ called the \textit{symmetry order algebra} \cite{evans2026operatoralgebraicapproachfusion, jones2025localtopologicalorderboundary}. If $L_{\phi}$ is Lagrangian, we say $\phi$ exhibits \textit{topological order}. We show that the symmetry order algebra is an invariant of states under stable equivalence \ref{thm:symmetryorderinvariance}, giving further evidence that stable equivalence may be the correct equivalence relation to describe universality classes.

The uniqueness part of our results is consistent with (though distinct from) previous kinematical classification results in the literature, particularly \cite{seifnashri2025disentanglinganomalyfreesymmetriesquantum, BolsDeRoeckDeWildeOliveira2025Classification}. These works provide a classification of $G$-actions by quantum cellular automata (QCA) on (finite-dimensional) tensor product Hilbert spaces, allowing the addition of ancilla with on-site $G$-symmetry. The conclusion is that G-actions by QCA are classified by the anomaly of the action, valued in $H^{3}(G,\text{U}(1))$ \cite{ChenGuLiuWen2013SPT, KapustinSopenko2025Anomalous}. In our language, an action of $G$ with anomaly $\omega\in H^{3}(G,\text{U}(1))$ is really an action of the fusion category $\text{Hilb}_{f.d.}(G,\omega)$, and thus their result states that there is a unique action of each of these pointed fusion categories on spin chains, up to the symmetric ancilla equivalence relation. While our mathematical set-ups are different, the ultimate similarity of results suggests that both approaches are accessing a single kinematical universality class for each symmetry category.

In a different direction, we obtain a result relevant to the program of topological holography in the sense of \cite{PhysRevB.107.155136}. There, it is argued that the collection of  operators localized near a physical boundary cut of a topologically ordered system should characterize the bulk topological order. This was mathematically formalized for spin systems, and shown in 2+1D to indeed be the case for string net models using the machinery of quasi-local algebras and DHR bimodules \cite{jones2024dhr, jones2025localtopologicalorderboundary, jones2025holographybulkboundarylocaltopological, ChuahHungarKawagoePenneysTombaWallickWei2024BoundaryAlgebras}. However, there can be many distinct but isomorphic quasi-local algebras for the same bulk topological order. How do these compare? This brings us to our third main result, paraphrased as a corollary to the above.

\bigskip

\bigskip

\textbf{Corollary.}

\medskip

\begin{enumerate}[start=3]
\item 
For any two string-net models, their boundary algebras are stably equivalent if and only if they have equivalent bulk topological order.
\end{enumerate}

\bigskip

We hope to extend this result from fusion spin chains, which are, up to bounded spread isomorphism, the boundary algebras of string net models, to a much larger class of ``non-chiral'' quasi-local algebras, realizing the vision of \cite{PhysRevB.107.155136}.

\bigskip

\subsection{Mathematical formulation}

For our mathematical formulation in infinite volume, there are two basic terminology translations from the above discussion: 

\medskip

\begin{enumerate}
    \item 
    Kinematical structures are encoded by \textit{quasi-local} C*-algebras with local subalgebras, rather than local Hilbert spaces \cite{MR1441540}.

    \medskip
    
    \item 
    Locality-preserving unitaries are formalized as \textit{bounded spread isomorphisms} between quasi-local algebras \cite{Schumacher2004Reversible}.
\end{enumerate}

\medskip

Quasi-local algebras are C*-algebras $A$ with unital subalgebras $A_{I}\subseteq A$ indexed by intervals in $\mathbbm{Z}$ encoding local operator structure (for our particular version of this concept, see Definition \ref{def:quasilocalalg}). When we start from a local Hilbert space picture, the local algebras $A_{I}$ should be interpreted physically as the algebra of operators on the local Hilbert space $H_{I}$. However, if the local Hilbert space is constrained (as is the case with anyon chains, for example), then the algebra $A_{I}$ is restricted to contain only the local operators which preserve boundary conditions, since for consistency reasons what we call local operators must be able to act on arbitrarily large patches of local Hilbert spaces. The physical properties we expect of the local Hilbert spaces are encoded in a handful of natural conditions on the local subalgebras borrowed from the algebraic quantum field theory (AQFT) framework in \cite{Haag1996Local}. The standard family of examples of quasi-local algebras, corresponding to spin chains, are those of the form $A=\otimes_{\mathbbm{Z}} M_{d}(\mathbbm{C})$, with $A_{I}:=\otimes_{I}M_{d}(\mathbbm{C})$ \cite{MR1441540}.

A bounded spread isomorphism $\alpha: A\rightarrow A^{\prime}$ between quasi-local algebras is a *-isomorphism between C*-algebras such that there is some $R\ge 0$ with $\alpha(A_{I})\subseteq A^{\prime}_{I^{+R}}$. In this paper we are using bounded spread isomorphisms to define a natural notion of equivalence of quasi-local algebras, rather than as a discrete dynamics (which is the ``quantum cellular automata'' picture \cite{Farrelly2020reviewofquantum}). In our ``passive" picture, one can think of a bounded spread isomorphism as defining a change of variables, where a new local variable is localized in a region ``uniformly not too far'' from the original local variable. If we have a bounded spread isomorphism $\alpha: A\rightarrow A^{\prime}$, then we have a canonical correspondence between structures encoding local physics (e.g. local Hamiltonians, states with correlation decay) on  $A$ and $A^{\prime}$.

Here we run into our first technicality: since we are going to add infinite-dimensional ancilla, all local algebras should be considered as W* (or von Neumann) algebras, rather than merely C*-algebras, though the quasi-local algebra itself is a C*-algebra and not a W*-algebra. W*-algebras have an additional weak*-topology, so all maps between quasi-local algebras we consider (e.g. bounded spread isomorphisms) should have the extra condition that they are locally continuous with respect to this topology; such maps are called \textit{locally normal} (see Definition \ref{def:locallynormal}). This is automatic when the local algebras do not contain a particular type of direct summand, which turns out to be a very mild constraint satisfied by all the quasi-local algebras in our paper. Nevertheless, this introduces some mathematical subtleties that are familiar from algebraic quantum field theory \cite{Araki1963Lattice, Haag1996Local}, but are not usually needed in the context of quantum spin systems.

We briefly recall the operator algebraic picture of fusion category symmetry in infinite volume introduced in \cite{evans2026operatoralgebraicapproachfusion}. If $\mathcal{C}$ is a fusion category, then a realization of $\mathcal{C}$ as symmetries on a quasi-local algebra $A$ is a choice of \textit{physical boundary subalgebra} $B\subseteq A$  (see Definition \ref{def:physicalboundarysubalg}), whose associated category of symmetry defects $\text{Sym}(B\subseteq A)$ (defined as DHR-bimodules with solitonic-type localization) is equivalent to $\mathcal{C}$. $B$ should be interpreted as the algebra of quasi-local symmetric operators. Indeed, the normalized fusion ring of $\mathcal{C}$ acts by unital completely positive (u.c.p.) maps on $A$ in a way such that $B$ is precisely the algebra of common fixed points, and thus it is natural to interpret the u.c.p. maps themselves as the symmetry operators and $B$ as the symmetric operators (the operators invariant under the symmetry). The inclusion $B\subseteq A$ is called a \textit{symmetry inclusion} if $B$ is a physical boundary subalgebra of A.

The canonical examples of quasi-local algebras with fusion category symmetry arise from anyon chains. Given a fusion category $\mathcal{C}$, anyon chains with $\mathcal{C}$-symmetry are parameterized by a choice of indecomposable right $\mathcal{C}$-module category $\mathcal{M}$, and a choice of strongly tensor generating object in the dual multi-fusion category $X\in \mathcal{C}^{*}_{\mathcal{M}}$. As expected, the quasi-local algebras associated to anyon chains are \textit{not} tensor products of matrix algebras, reflecting the fact that the local Hilbert spaces have no local tensor product decomposotion in general. These quasi-local algebras nevertheless have trivial DHR bimodules, hence should count as reasonable ``standalone'' quasi-local algebras (rather than as a boundary algebra of some TQFT). The physical boundary subalgebras of anyon chains have been called \textit{fusion spin chains} in the literature, and have been studied extensively in the context of subfactor theory for over 40 years \cite{MR1642584}. More recently, they have been explicitly studied as abstract quasi-local algebras from the non-invertible symmetries viewpoint \cite{jones2024dhr, jones2025quantumcellularautomatacategorical, JoLi24}.

In order to give a precise statement of the first main result we consider the quasi-local algebra $\mathcal{L}^{\infty}:=\overline{\otimes}_{\mathbbm{Z}} \mathcal{B}(\ell^{2}(\mathbbm{N}))$, which we can think of as the infinite-volume limit of a chain of harmonic oscillators (or if you prefer the local Hilbert space $L^{2}(S^{1})$, a chain of rotors). Our first subtlety appears here: the tensor product is the W*-tensor product, denoted by $\overline{\otimes}$, rather than the usual C*-tensor product. If $A$ is an arbitrary quasi-local algebra, we can then define the \textit{stabilization} $A^{\infty}$. If $A$ is locally finite-dimensional (and therefore nuclear), which happens to be the situation we are interested in, this is simply the (unique) tensor product of C*-algebras $A^{\infty}=A\otimes \mathcal{L}^{\infty}$, with obvious local algebra structure $A_{I}\otimes \mathcal{L}^{\infty}_{I}$. If the local algebras are infinite-dimensional W*-algebras, however, then we have to take the W*-tensor product locally, i.e. $A^{\infty}_{I}:=A_{I}\overline{\otimes} \mathcal{L}^{\infty}_{I}$, then take the C*-inductive limit of these.

We can now state the formal version of the first main result.

\begin{thm}(see Theorem \ref{thm:MoritaEquivimpliesStableEquiv})
If $A$ is the quasi-local algebra of an anyon chain, then there is a bounded spread isomorphism $A^{\infty}\cong \mathcal{L}^{\infty}$. 
\end{thm}

To state the second main result, we recall the obvious (and rather strong) notion of equivalence of symmetry inclusions. An \textit{equivalence} of symmetry inclusions $B_{1}\subseteq A_{1}$ and $B_{2}\subseteq A_{2}$ is a bounded spread isomorphism $\alpha: A_{1}\rightarrow A_{2}$ such that $\alpha(B_{1})=B_{2}$. If $B_{1}\subseteq A_{1}\cong B_{2}\subseteq A_{2}$, then the symmetry categories $\text{Sym}(B_1\subseteq A_{1})\cong \text{Sym}(B_{2}\subseteq A_{2})$. Thus it is natural to fix a fusion category $\mathcal{C}$ and  ask for equivalences of symmetry inclusions realizing $\mathcal{C}$. It is this that we will use as a basic (and rather strong) equivalence relation between realizations of a fusion category $\mathcal{C}$. As mentioned above, this is quite restrictive, so we introduce the more general definition of \textit{stable equivalence}: we say two symmetry inclusions $B_{1}\subseteq A_{1}$ and $B_{2}\subseteq A_{2}$ are \textit{stably equivalent} if $B^{\infty}_{1}\subseteq A^{\infty}_{1}\cong B^{\infty}_{2}\subseteq A^{\infty}_{2}$. We have the following result

\begin{thm} (see Theorem \ref{thm:stableequivsymm}) Let $\mathcal{C}$ be a unitary fusion category. Any two anyon chains realizing $\mathcal{C}$-symmetry are stably equivalent.
\end{thm}

We note a subtle point: it is not obvious that if $B\subseteq A$ is a symmetry inclusion, then $B^{\infty}\subseteq A^{\infty}$ is also a symmetry inclusion. However, we show that this is indeed the case for anyon chains (see Theorem \ref{thm:symmetryorderinvariance}).

Finally, our result concerning stable equivalence of boundary algebras is phrased in terms of fusion spin chains. Given a unitary indecomposable multi-fusion category $\mathcal{C}$ and a strong tensor generating object $X\in \mathcal{C}$, there is an associated quasi-local algebra $A(\mathcal{C},X)$ called a \textit{fusion spin chain}. These are precisely the boundary algebras of Levin-Wen type string-net models \cite{jones2025localtopologicalorderboundary}, whose bulk topological order is described by $\mathcal{Z}(\mathcal{C})$ \cite{LevinWen2005StringNet, BolsKjaer2025LevinWenI, BolsKjaer2026LevinWenII}. We have the following theorem.

\begin{thm} (see Corollary \ref{cor:fusionspinchains})
    If $A(\mathcal{C},X)$ and $A(\mathcal{D},Y)$ are fusion spin chains, then they are stably equivalent if and only if $\mathcal{Z}(\mathcal{C})\cong \mathcal{Z}(\mathcal{D})$ as unitary braided fusion categories.
\end{thm}

\bigskip

\noindent \textbf{Acknowledgements.} The authors would like to thank:
Clement Delcamp, David Penneys, Ryan Thorngren, Sahand Seifnashri, Wilbur Shirley, Nikita Sopenko, Xiao-Gang Wen, Dominic Williamson, and Xinping Yang for helpful conversations. Both authors were supported by NSF DMS-2247202. ChatGPT Sol 5.6 was used to help make diagrams and editorial improvements, and to assist in strengthening our results from a previous version of this paper. In particular, a preliminary version of Lemma 5.1 was proposed by ChatGPT Sol 5.6. All mistakes are the authors'.

\section{Categorical symmetries on the 1D lattice}

While nets of algebras are ubiquitous throughout the literature, there are many minor variations in definition and notation so we take the time to set our notation and assumptions. We use the terminology of \textit{quasi-local algebra} following \cite{MR887100, MR1441540}, which is more convenient notationally. We also refer to these references for the basic notions of C*-algebras, W*/von Neumann algebras, and the various topologies involved.

The set of integers $\mathbb{Z}$ forms a discrete metric space with  the standard metric, and a natural order is imposed on the set of bounded intervals by the subset relation. We let $\text{Int}(\mathbbm{Z})$ denote the poset of bounded intervals.

\begin{defn}\label{def:quasilocalalg}
    
A quasi-local algebra over $\mathbbm{Z}$ consists of a unital C*-algebra $A$ together with an assignment, for each bounded interval $I\subseteq \mathbbm{Z}$, a W* subalgebra $A_I\subseteq A$ with separable pre-dual such that

\begin{enumerate}
\item 
$A_{\varnothing}=\mathbbm{C}1$,
\item
(Isotony) If $I\subseteq J$, then $A_{I}\subseteq A_{J}$ (and the inclusion is normal),
\item
(Locality) If $I\cap J=\emptyset$, then $[A_{I},A_{J}]=0$,
\item
(Non-degeneracy) $\bigcup_{I} A_{I}$ is norm dense in $A$.
\end{enumerate}  
\end{defn}

\noindent These will be the most general conditions required of a quasi-local algebra. However, there is an additional condition that is very important in practice, first introduced in this context in \cite{jones2024dhr}:

\bigskip

\begin{enumerate}[start=5]
\item 
(\textit{Weak algebraic Haag duality) There is some $R\ge 0$ such that $\{a\in A\ :\ [a,A_{I^{c}}]=0\}\subseteq A_{I^{+R}}$. If $R=0$, we say $A$ satisfies algebraic Haag duality.} 
\end{enumerate}

\bigskip

The notation $I^{+R}$ denotes the $R$-ball of the interval $I$, rather than translation. Given an unbounded region $F$ (e.g. the complement of a bounded interval $I^{c}$ or a left or right unbounded interval $I_{\le n}$ or $I_{\ge n}$), we define $A_{F}$ to be the \textit{C*-algebra} generated by the W*-algebras $A_{I}$, with $I\subseteq F$.

\begin{defn}\label{def:locallynormal} If $A$ and $B$ are both C*-algebras each equipped with a norm dense family of W*-subalgebras $\{A_{I}\subseteq A\}_{I\in \Lambda}$ and $\{B_{J}\subseteq B\}_{J\in \Gamma}$, a linear map $\Phi:A\rightarrow B$ is called \textit{locally normal} if for every $I\in \Lambda$, $\Phi(A_{I})\subseteq B_{J}$ for some $J\in \Gamma$, and $\Phi|_{A_{I}}: A_{I} \rightarrow B_{J}$ is normal.
\end{defn}

For quasi-local algebras $A$ over $\mathbbm{Z}$, if we say a map is locally normal we will always mean with respect to the family of von Neumann algebras $\{A_{I}\}_{I\in \text{Int}(\mathbbm{Z})}$. We also note that in many cases of interest, local normality is automatic for $*$-homomorphisms. Indeed, we introduce the following condition that ensures local normality for homomorphisms is automatic.

\begin{thm}\cite[Theorem 6.5]{BaudierBragaFarahVignatiWillett2024Embeddings} Any *-homomorphism from a separable von Neumann algebra $M$ with no direct summands of the form $M_{k}(L)$, where $k\ge 0$ and $L$ is an infinite-dimensional abelian von Neumann algebra, to a separable von Neumann algebra $N$, is automatically normal.
\end{thm}

\bigskip

\noindent This leads us to define the following condition on quasi-local algebras which will be very convenient.

\begin{enumerate}[start=6]
\item 
(\textit{Local normality}) For each $I\subseteq \mathbbm{Z}$, the W*-algebra $A_{I}$ has no summand of the form $M_{k}(L)$, where $k\ge 0$ and $L$ is an infinite-dimensional abelian von Neumann algebra. 
\end{enumerate}

In particular, if $A_{I}$ is a finite direct sum of factors, a countably infinite direct sum of infinite factors, or a countable direct sum of finite factors where each dimension of the finite factor summands occurs finitely often, then $A$ will satisfy this condition. In this paper, all the quasi-local algebras we consider will have local algebras that are a finite direct sum of factors, hence will satisfy condition $6$.

\begin{defn} Let $A$ and $B$ be quasi-local algebras satisfying conditions $(1)-(6)$. A \textit{bounded spread isomorphism} is a C*-algebra isomorphism $\alpha:A\rightarrow B$ such that there is an $R\ge 0$ with $\alpha(A_{I})\subseteq B_{I^{+R}}$.
\end{defn}

Condition $(5)$ guarantees that $\alpha^{-1}$ also has bounded spread, and $(6)$ guarantees that $\alpha|_{A_{I}}\rightarrow B_{I^{+R}}$ is normal, so that $\alpha$ is locally normal.

\subsection{Discrete tensor categories}

Recall that a W*-category is a unitary Cauchy complete C*-category, such that every morphism space has a predual as a Banach space \cite{MR808930, MR3687214}. In particular, a W*-category is closed under taking summands and direct sums. A W*-tensor category is a W*-category which is monoidal, such that tensoring with morphisms is separately weak*-continuous in each tensor factor \cite{MR4419534}. We assume all our W*-algebras are separable (i.e. have separable pre-dual).

\begin{defn}
A W*-category is \textit{discrete} if for every object $X\in \mathcal{D}$, the endomorphism algebra is a countable $\ell^{\infty}$ direct product of type $I$ factors, i.e.

$$\text{End}_{\mathcal{D}}(X)\cong \prod_{i} \mathcal{B}(H_{i}).$$
\end{defn}

\noindent This notion of discreteness is a generalization of the concept of $\textit{semisimple}$, which occurs precisely when all endomorphism algebras are finite-dimensional. Note that every C*-tensor category such that all endomorphism algebras are finite-dimensional is automatically W* and semisimple, hence discrete in the above sense. The main examples we have in mind which are discrete rather than just semisimple are of the form $\text{Hilb}(\mathcal{C})$, where $\mathcal{C}$ is an ordinary semisimple C*-tensor category (possibly without simple units or duals).

 Following \cite{MR3687214, MR3509018}, the idea is that we allow formal ``Hilbert space'' direct sums of objects. More precisely, any finitely semisimple C*-category $\mathcal{C}$ is equivalent to the category of $\dagger$ functors $\text{Fun}[\mathcal{C}^{op}, \text{Hilb}_{f.d.}]$, via the Yoneda embedding $a\mapsto \text{Hom}_{\mathcal{C}}(\ \cdot\ ,\ a)$. If $\mathcal{C}$ has a C*-tensor product structure, this is visible in the Yoneda picture via \textit{Day convolution}. In any case, the completion we care about, $\text{Hilb}(\mathcal{C})$, can be viewed as simply replacing the target category $\text{Hilb}_{f.d.}$ with $\text{Hilb}$, where by the latter we actually mean \textit{countably infinite-dimensional} Hilbert spaces. We still use Day convolution for the tensor product and have a naturally defined unitary associator, but the resulting category is a W*-tensor category. Intuitively, this allows for ``Hilbert space-enriched'' direct sums of objects. We will give a more concrete presentation of this category that may be more useful to readers, that takes seriously the idea of formally allowing Hilbert space multiplicities.

An object of $\text{Hilb}(\mathcal{C})$ can be written formally as a (in general infinite) direct sum $\bigoplus_{i} H_{i}\boxtimes a_{i}$ (with $a_i\in \mathcal{C})$, where $\boxtimes$ is a formal symbol. If $\mathcal{C}$ is finitely semisimple (which we will assume here to obtain the local semi-simplicity condition) this direct sum can always be taken to have finitely many summands.

Then $$\text{Hom}_{\text{Hilb}(\mathcal{C})}
\left(\bigoplus_{i} H_{i}\boxtimes a_{i}, \bigoplus_{j} K_{j}\boxtimes b_{j}\right):=\bigoplus_{i,j} \mathcal{B}(H_{i}, K_{j})\otimes_{\mathbbm{C}} \text{Hom}_{\mathcal{C}}(a_{i},b_{j}),$$

\noindent with the obvious composition and $\dagger$. The tensor product is given by

$$\left(\bigoplus_{i} H_{i}\boxtimes a_{i}\right)\otimes \left(\bigoplus_{j} K_{j}\boxtimes b_{j}\right):= \bigoplus_{i,j} (H_{i}\otimes K_{j})\boxtimes (a_{i}\otimes b_{j}),$$

\noindent and the associator is just applied to the tensor factor of objects in $\mathcal{C}$ (ignoring the $\boxtimes$). By construction, we see that $H\boxtimes (\oplus_{i} a_{i})\cong \oplus_{i} H\boxtimes a_{i}$. Furthermore, it is easy to show that every object $a\in \text{Hilb}(\mathcal{C})$ is unitarily isomorphic in this category to 

$$\bigoplus_{b\in \text{Irr}(\mathcal{C})} \text{Hom}_{\text{Hilb}(\mathcal{C})}(b, a)\boxtimes b,$$

\noindent where $\text{Irr}(\mathcal{C})$ is a set of representatives of isomorphism classes of simple objects, and $\text{Hom}_{\text{Hilb}(\mathcal{C})}(b, a)$ is the Hilbert space given by the composition inner product $\langle f\ |\ g\rangle 1_{b}:=f^{\dagger}\circ g$, which is well-defined since $b$ is simple.

We also record the following observation. If $\mathcal{C}$ is discrete and closed under countable W*-direct sums, then the inclusion $\mathcal{C}\hookrightarrow \text{Hilb}(\mathcal{C})$ is an equivalence.

\subsection{Categorical quasi-local algebras}

In this section, we recall the construction of nets of algebras from a W*-tensor category, extending the definitions of  \cite{jones2024dhr}. Let $\mathcal{D}$ be a W*-tensor category, and $X\in \mathcal{D}$ be an object such that $X^{\otimes n}\ne 0$ for all $n$.

We now describe a quasi-local algebra over $\mathbbm{Z}$, 
$$A:=A(\mathcal{D}, X).$$

\noindent For any interval $I\in \text{Int}(\mathbbm{Z})$, the local algebra is defined as $A^{0}_{I}:=\text{End}_{\mathcal{D}}(X^{\otimes I})$. 
Any $f\in A^{0}_{I}$ can be represented in the graphical calculus by

$$
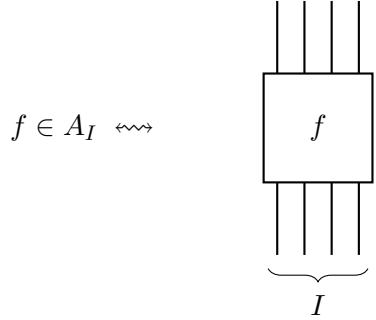
\begin{tikzpicture}[scale=1.2]

  \node at (-2.6,0) {$f\in A_I \;\leftrightsquigarrow$};

  \draw[thick] (-0.6,0.6) rectangle (0.6,-0.6);
  \node at (0,0) {$f$};

  \draw[thick] (-0.45,1.4) -- (-0.45,0.6);
  \draw[thick] (-0.15,1.4) -- (-0.15,0.6);
  \draw[thick] (0.15,1.4) -- (0.15,0.6);
  \draw[thick] (0.45,1.4) -- (0.45,0.6);

  \draw[thick] (-0.45,-0.6) -- (-0.45,-1.4);
  \draw[thick] (-0.15,-0.6) -- (-0.15,-1.4);
  \draw[thick] (0.15,-0.6) -- (0.15,-1.4);
  \draw[thick] (0.45,-0.6) -- (0.45,-1.4);

  \draw[decorate,decoration={brace,mirror,amplitude=5pt}]
    (-0.55,-1.55) -- (0.55,-1.55)
    node[midway,below=6pt] {$I$};

\end{tikzpicture}
$$

For $I\subseteq J$, the maps $\pi_{I,J}:A^{0}_{I}\hookrightarrow A^{0}_{J}$ are defined by tensoring with $1_{X}$ at all sites in $J\setminus I$, which we write as $f\mapsto 1_{X}\otimes \dots 1_{X} \otimes f \otimes 1_{X}\otimes \dots 1_{X}$. In the graphical calculus this is

$$
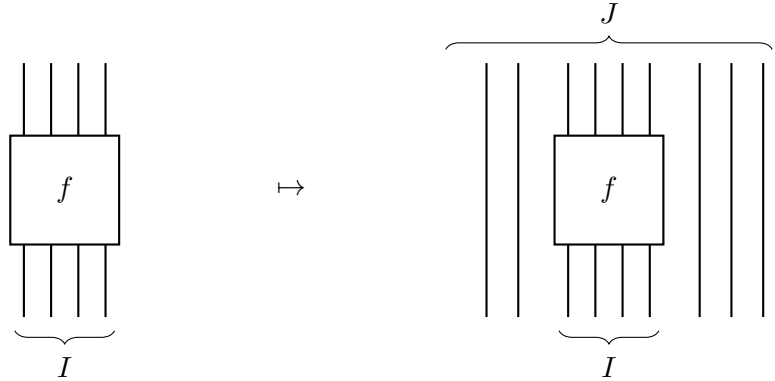
\begin{tikzpicture}[scale=1.2]

  \draw[thick] (-0.6,0.6) rectangle (0.6,-0.6);
  \node at (0,0) {$f$};

  \draw[thick] (-0.45,1.4) -- (-0.45,0.6);
  \draw[thick] (-0.15,1.4) -- (-0.15,0.6);
  \draw[thick] (0.15,1.4) -- (0.15,0.6);
  \draw[thick] (0.45,1.4) -- (0.45,0.6);

  \draw[thick] (-0.45,-0.6) -- (-0.45,-1.4);
  \draw[thick] (-0.15,-0.6) -- (-0.15,-1.4);
  \draw[thick] (0.15,-0.6) -- (0.15,-1.4);
  \draw[thick] (0.45,-0.6) -- (0.45,-1.4);

  \draw[decorate,decoration={brace,mirror,amplitude=5pt}]
    (-0.55,-1.55) -- (0.55,-1.55)
    node[midway,below=6pt] {$I$};

  \node at (2.5,0) {$\mapsto$};

  \draw[thick] (5.4,0.6) rectangle (6.6,-0.6);
  \node at (6,0) {$f$};

  \draw[thick] (5.55,1.4) -- (5.55,0.6);
  \draw[thick] (5.85,1.4) -- (5.85,0.6);
  \draw[thick] (6.15,1.4) -- (6.15,0.6);
  \draw[thick] (6.45,1.4) -- (6.45,0.6);

  \draw[thick] (5.55,-0.6) -- (5.55,-1.4);
  \draw[thick] (5.85,-0.6) -- (5.85,-1.4);
  \draw[thick] (6.15,-0.6) -- (6.15,-1.4);
  \draw[thick] (6.45,-0.6) -- (6.45,-1.4);

  \draw[thick] (4.65,1.4) -- (4.65,-1.4);
  \draw[thick] (5.00,1.4) -- (5.00,-1.4);

  \draw[thick] (7.00,1.4) -- (7.00,-1.4);
  \draw[thick] (7.35,1.4) -- (7.35,-1.4);
  \draw[thick] (7.70,1.4) -- (7.70,-1.4);

  \draw[decorate,decoration={brace,amplitude=5pt}]
    (4.20,1.55) -- (7.80,1.55)
    node[midway,above=6pt] {$J$};

  \draw[decorate,decoration={brace,mirror,amplitude=5pt}]
    (5.45,-1.55) -- (6.55,-1.55)
    node[midway,below=6pt] {$I$};

\end{tikzpicture}
$$

\noindent Using the assumption on $X$ (and the fact that this is a W*-tensor category), this gives us a normal map of von Neumann algebras. The colimit yields the quasi-local C*-algebra

$$A(\mathcal{D},X):=\text{colim}_{I} A^{0}_{I},$$

\noindent with local algebras $A(\mathcal{D},X)_{I}$ the images of $\pi_{I,\infty}(A^{0}_{I})\subseteq A(\mathcal{D},X)_{I}$ in the colimit. 

This gives us a net of C*-algebras, but if we are interested in a net of W*-algebras, we would be in slight trouble: it is not obvious that the kernel $\pi_{I,\infty}$ is a W*-ideal, and in particular that $A(\mathcal{D},X)_{I}$ as we have defined it is indeed a W*-algebra.

For an arbitrary W*-category, this would be slightly problematic, since it is unclear that the image under $\pi_{I,\infty}$ would be a von Neumann algebra. However, since $A^{0}_{I}\cong \prod_{i\in Z} \mathcal{B}(H_{i})$, then any W*-ideal is given by specifying a subset $Y\subseteq Z$, in which case $I_{Y}=\{(a_{i})_{i\in Z}\ : a_{i}=0\ \text{if}\ i\in Y \}$.

Then for each $I\subseteq J$, $\text{Ker}(\pi_{I,J})=I_{Y_{J}}$, hence setting $W=\bigcup_{I\subseteq J} Y_{J}$, we have $\text{Ker}(\pi_{I,\infty})=I_{W}$, which is a W*-ideal, hence $A(\mathcal{D}, X)_{I}$ is normal.
We note, however, that in nearly all of our examples, the maps $\pi_{I,J}$ are already injective.

At this level of generality, $A(\mathcal{D},X)$ satisfies the conditions $(1)$-$(4)$, but may fail $(5)$ and $(6)$.

\subsection{Fusion spin chains}

A unitary fusion category is a semisimple W*-category with finitely many isomorphism classes of simple objects, such that every object has a dual object and the unit $\mathbbm{1}$ is simple \cite{MR2183279, MR3242743, MR1444286}. The standard examples are $\text{Rep}(G)$ and $\text{Hilb}_{f.d.}(G, \omega)$. If we drop the condition that the unit is simple, we obtain the definition of a unitary multi-fusion category. A unitary multifusion category $\mathcal{C}$ is called \textit{indecomposable} if for any two summands $\mathbbm{1}_{i}, \mathbbm{1}_{j}$ of $\mathbbm{1}$, $\mathbbm{1}_{i} \mathcal{C} \mathbbm{1}_{j}\ne 0$. Recall an object $X\in \mathcal{C}$ is a \textit{strong tensor generator} if there is some $n\ge 0$ such that all simple objects in $\mathcal{C}$ are isomorphic to summands of $X^{\otimes n}$. Note that we need a single $n$ containing all isomorphism classes of simples simultaneously, making this notion stronger than just ``tensor generating".

\begin{defn}
If $\mathcal{D}$ is an indecomposable unitary multi-fusion category and $X\in \mathcal{D}$ is a strong tensor generator, then the quasi-local algebra $A(\mathcal{D},X)$ is called a \textit{fusion spin chain}.
\end{defn}

The strong tensor generating condition is not a particularly serious restriction. If $X$ is not strong tensor generating for $\mathcal{D}$, then since $X^{\otimes n} \ne 0$ for all $n$ , then a summand of $\mathbbm{1}$ will be contained in some sufficiently large tensor power of $X$. Thus we can take the largest full (but not necessarily unital in the multi-fusion case) tensor subcategory $\mathcal{D}^{\prime}$ that appears under a single tensor power of $X$, say $X^{\otimes n}$. After coarse-graining the net into intervals of length $n$, we see that the new net $A(\mathcal{D}^{\prime}, X^{\otimes n})$ satisfies our hypothesis. Thus, this condition is a kind of non-degeneracy condition, ensuring that we have calibrated our length scale appropriately so that the category $\mathcal{D}$ is the one containing the correct information.

Fusion spin chains satisfy (strong) algebraic Haag duality hence condition $(5)$ \cite{jones2024dhr}, and also local normality $(6)$.

We point out one particular special case which plays a major role in our paper. Recall that an indecomposable multi-fusion category is called \textit{Morita trivial} if it is Morita equivalent to $\text{Hilb}_{f.d.}$. Such multi-fusion categories are unitarily equivalent to $\text{Mat}_{n}(\text{Hilb}_{f.d.})$. This category in turn is equivalent to the W*-tensor category $\text{End}(\mathcal{M})$, where $\mathcal{M}$ is a semisimple W*-category of rank $n$.

\begin{defn}
If $\mathcal{E}\cong \text{Mat}_{n}(\text{Hilb}_{f.d.})$ and $X\in \mathcal{E}$ a strong tensor generator, then $A(\mathcal{E},X)$ is called an \textit{anyon chain}.
\end{defn}

The terminology ``anyon chain" is not a very good choice, specifically because there are no anyons (in the usual sense of 2D topological point defects/quasi-particle excitations) anywhere in sight. In the original study of anyon chains, $\mathcal{E}=\text{End}(\text{Fib})$, and $\text{Fib}$ itself is realized as anyons of some 2D topologically ordered system, which partly explains the origins  of the term \cite{Feiguin2007}. Nevertheless, this vocabulary seems to be entrenched and thus we feel obligated to use it for clarity of communication.

A better way to think about anyon chains is in terms of \textit{graphs}. If $\mathcal{E}=\text{Mat}_{n}(\text{Hilb}_{f.d.})$, then $X=(X_{ij})_{1\le i,j\le n} \in \mathcal{E}$ is determined up to isomorphism by specifying the dimension of the Hilbert spaces $N_{ij}=\text{dim}(X_{ij})$. Then define an oriented graph $\mathcal{G}_{X}$ with $n$ vertices, ordered $1$  to $n$, and $N_{ij}$ edges from vertex $i$ to vertex $j$. Then we see immediately that $A(\mathcal{E},X)_{I}=\text{End}_{\mathcal{E}}(X^{\otimes I})$ can be identified with the space of linear operators on the Hilbert space $\mathbbm{C}[P_{m}]$, with an orthonormal basis consisting of paths of length $m$, which preserve source and target of paths. This is discussed extensively in \cite{jones2025quantumcellularautomatacategorical}. We discuss this point of view in more detail in \ref{hilbspacelevel}.

\subsection{DHR bimodules}

From this section onward, we assume our quasi-local algebras satisfy conditions $(1)-(6)$.

If $A$ is a C*-algebra, an $A$-$A$ correspondence is an $A$-$A$ bimodule $X$, with a right $A$-valued inner product $\langle \cdot\ |\ \cdot \rangle: X\times X\rightarrow A$, making $X_{A}$ into a right Hilbert C*-module over $A$, with a left action of $A$ by adjointable operators. For formal definitions, see \cite{MR1325694} or \cite{MR4419534} for a more categorical description. If $X$ is a correspondence, a projective basis is a finite subset $\{b_{i}\}^{n}_{i=1}$ such that for all $x\in X$, $x=\sum_{i} b_{i}\langle b_{i}\ |\ x\rangle$

\begin{defn}
An $A$-$A$ correspondence $X$ is a \textit{DHR bimodule} if for every sufficiently large interval $I$, there exists a finite projective basis $\{b_{i}\}^{n}_{i=1}$ such that for each $i$, $[A_{I^{c}},b_{i}]=0$
\end{defn}

Since $A$ satisfies weak algebraic Haag duality, this category was shown to be a braided C*-tensor category $\text{DHR}(A)$\ \cite{jones2024dhr}. In our setting of local von Neumann algebras, we can upgrade this to a W*-tensor category.

\begin{lem}
For any quasi-local algebra satisfying the conditions local normality and weak algebraic Haag duality, $\text{DHR}(A)$ is a W*-tensor category.
\end{lem}

\begin{proof}
First we show that $\text{DHR}(A)$ is a W*-category. It suffices to show that for any $X,Y\in \text{DHR}(A)$, $\text{End}(X\oplus Y)$ is a W*-algebra. Let $R$ be the max of the localization radii for $X$ and $Y$, and choose an interval $I\ge R$, and localized projective bases $\{b_{i}\}\subseteq X$ and $\{c_{j}\}\subseteq Y$ respectively. Then $E:=\{(b_{i},0)\}\cup \{(0,c_{j})\}$ is a projective basis localized in $I$ for $X\oplus Y$, and thus defines and amplimorphism $\pi: A\rightarrow M_{n}(A)$ for $n=|E|$. For any bimodule endomorphism $T\in \text{End}(X\oplus Y)$, relabelling the elements of $E$ by $\{e_i\}$, $T$ is uniquely determined by an element $\widetilde{T}=(\langle e_{i}\ | T(e_j)\rangle)_{i,j}\in PM_{n}(A)P$, where $P=(\langle e_{i} |\ e_j\rangle)_{i,j}\in M_{n}(A)$ is an orthogonal projection. But by weak algebraic Haag duality, $P\in M_{n}(A_{I^{+R'}})$ for some $R'$. In fact, we see that in the amplimorphism picture, $\text{End}(X\oplus Y)\cong \{a\in PM_{n}(A_{I^{+R'}})\ : a\pi(b)=\pi(b)a,\ b\in A\}$. But this equation is weak*-closed in the von Neumann algebra $PM_{n}(A_{I^{+R'}})P$, hence has a canonical pre-dual. 

Using this type of argument, it is straightforward to check that composition and $*$ are (separately) weak* continuous. Local normality is used to show that tensoring (which in the amplimorphism picture involves applying an amplimorphism to an intertwiner, which is local) is separately weak*-continuous.

\end{proof}

Following \cite{evans2026operatoralgebraicapproachfusion}, let $B\subseteq A$ be a unital subalgebra. Then $B$ is called a \textit{quasi-local subalgebra} if $B_{I}:=B\cap A_{I}$ is a W*-algebra, and $\vee_{I} B_{I}$ is norm dense in $B$. A quasi-local subalgebra satisfies \textit{relative weak Haag duality} if there exists an $R\ge 0$ such that $\{a\in A\ :\ [a,B_{I^{c}}]=0\}\subseteq A_{I^{+R}}$. If $R=0$, we say $A$ satisfies \textit{ algebraic Haag duality} relative to $B$. We will also assume that $B$ satisfies local normality.

We point out that if $B\subseteq A$ satisfies relative weak algebraic Haag duality then this implies that $B$ itself satisfies weak algebraic Haag duality. Furthermore, we interpret $\varnothing^{+R}=\varnothing$, and thus weak relative algebraic Haag duality implies the inclusion $B\subseteq A$ is irreducible. In particular, this implies $A$ and $B$ both have trivial center. Irreducibility also implies that if there exists a conditional expectation $E:A\rightarrow B$, it is unique.

\begin{defn}{\cite{evans2026operatoralgebraicapproachfusion}}
If $A$ is a quasi-local algebra over $\mathbbm{Z}$, a \textit{local subalgebra} is a quasi-local subalgebra $B\subseteq A$ satisfying weak relative algebraic Haag duality, such that there exists a locally normal conditional expectation $E:A\rightarrow B$, and for every sufficiently large interval $I$, there is a finite subset $\{b_{i}\}\subseteq A_{I}$ such that $a=\sum_{i}b_{i}E(b^{\dagger}_{i}a)$.
\end{defn}

Recall that for a conditional expectation $E:A\rightarrow B$, we can form the $B$-$B$ correspondence $_{B} A_{B}:=A$ as a vector space, with left and right actions $b_{1}\cdot a\cdot b_{2}=b_{1}ab_{2}$. The right $B$-valued inner product is given by 

$$\langle a\ |\ a^{\prime}\rangle:=E(a^{\dagger}a^{\prime}).$$

\noindent Then by \cite{MR4419534}, $_{B}A_{B}\in \text{Corr}(B)$ is a dual Q-system, in the sense of \cite{evans2026operatoralgebraicapproachfusion}.

\bigskip

There are two significant facts concerning local subalgebras:

\begin{enumerate}
    \item 
    $\text{DHR}_{0}(B)$ is a unitary braided tensor category in the sense of \cite{2021arXiv211106378C}.

\medskip
    
    \item 
    The dual, connected Q-system $_{B}A_{B}\in \text{DHR}_{0}(B)$ and is commutative.
\end{enumerate}

\medskip

\noindent If $L$ is a commutative Q-system in a braided unitary fusion category, then we may equip the category of right $L$-modules $\mathcal{B}_{L}$ with the $\beta_{-}$ monoidal structure as in \cite{MR1729094, evans2026operatoralgebraicapproachfusion}. Then the free module functor can be canonically equipped with a central structure $\mathcal{B}\rightarrow \mathcal{Z}(\beta_{-}(\mathcal{B}_{L}))$. $L$ is called \textit{Lagrangian} if this central functor is an equivalence.

\begin{defn}\label{def:physicalboundarysubalg}
A local subalgebra $B\subseteq A$ is called a \text{physical boundary subalgebra} if $\text{DHR}(B)$ is the Hilb completion of a fusion category (necessarily $\text{DHR}_{0}(B)$), and the dual Q-system $_{B}A_{B}\in \text{DHR}_{0}(B)$ is Lagrangian. We call such an inclusion a \textit{symmetry inclusion}, and we call $\beta_{-}(\text{DHR}_{0}(B)_{A})$ the \textit{symmetry fusion category}.
\end{defn}

We can understand the symmetry category directly as a tensor subcategory of correspondences of $A$ satisfying a solitonic type of localization condition, $\text{DHR}_{-}(A|B)_{0}$ \cite{evans2026operatoralgebraicapproachfusion}. Briefly recall that $X\in \text{DHR}_{-}(A|B)_{0}$ if for all sufficiently large $I$, there exists a locally normal projective basis $\{b_{i}\}$ for $X$ such that $[A_{>I},b_{i}]=0$ and $[B_{< I}, b_{i}]=0$. Here $\text{DHR}{-}(A|B)_{0}$ denotes the dualizable objects in this category. Slightly modifying the argument of \cite{evans2026operatoralgebraicapproachfusion} to take dualizability into account (which is easy since $_{B}A_{B}$ is dualizable), we see that $\beta_{-}(\text{DHR}_{0}(B)_{A})\cong \text{DHR}_{-}(A|B)_{0}$.

\begin{defn}
If $\mathcal{C}$ is a unitary fusion category, a realization of $\mathcal{C}$ as symmetries on a quasi-local algebra $A$ is a symmetry inclusion $B\subseteq A$ such that $\mathcal{C}\cong \beta_{-}(\text{DHR}_{0}(B)_{A})$. 
\end{defn}

Of course, if there is an equivalence $\mathcal{C}\cong \beta_{-}(\text{DHR}_{0}(B)_{A})$, there is in general more than one (these will form a torsor over the group of monoidal auto-equivalences $\text{Aut}_{\otimes}(\mathcal{C})$, but we don't care \textit{which} one so we don't keep track of this data).

\begin{defn}
Two symmetry inclusions $B\subseteq A$ and $B^{\prime}\subseteq A^{\prime}$ are \textit{equivalent} if there exists a bounded spread isomorphism $\alpha: A\rightarrow A^{\prime}$ such that $\alpha(B)=B^{\prime}$. An equivalence of two realizations of $\mathcal{C}$ is an equivalence of the corresponding symmetry inclusions.
\end{defn}

Note that this is a coherent definition of symmetry equivalence: if we have an equivalence of symmetry inclusions, conjugating by $\alpha$ will give an equivalence 

$$\text{DHR}_{-}(A|B)_{0}\cong \text{DHR}_{-}(A^{\prime}|B^{\prime})_{0}.$$

\noindent We do not require that this natural equivalence intertwine our particular choices of \textit{identification} of $\mathcal{C}$ with $\text{DHR}_{-}(A|B)_{0}$ and $\text{DHR}_{-}(A^{\prime}|B^{\prime})_{0}$, respectively. If we have a fixed symmetry realization and we are attempting to classify symmetric phases with respect to this, such an equivalence intertwining specific identifications would be appropriate. The relation we have defined above is more suitable to kinematic equivalence.

\subsection{DHR theory for fusion spin chains}

The category of DHR bimodules for fusion spin chains $A(\mathcal{D},X)$ can be calculated in terms of the Drinfeld center of the category $\mathcal{D}$. We briefly recall a formal definition of this category in the graphical calculus, so that we can make use of these pictures when describing the equivalence. We note that we are using the conventions in \cite{MR1966525, Izumi2000LongoRehrenI} rather than the conventions in \cite{MR3242743}.

Let $(Z,\sigma)\in \mathcal{Z}(\mathcal{D})$, where $\sigma=\{\sigma_{Z,a}: Z\otimes a\cong a\otimes Z\}_{a\in \mathcal{D}}$, natural in $a$. Graphically, we represent this by

\begin{center}
\begin{minipage}[c]{0.1cm}
    $$\sigma_{Z,a}:=$$
\end{minipage}
\begin{minipage}[c]{0.3\textwidth}
    $$
\begin{tikzpicture}[scale=0.7]
  
  \draw[black, line width=1.2pt] (0,2) .. controls (0,1.2) and (1,0.8) .. (1,0);

  \draw[red, line width=1.2pt] (1,2) .. controls (1,1.2) and (0,0.8) .. (0,0);

  \node[below] at (0,0) {$Z$};
  \node[below] at (1,0) {$a$};
\end{tikzpicture}
$$
\end{minipage}
\end{center}

\

\noindent Naturality is the condition that for any $f\in \mathcal{D}(a,b)$, we have $(f\otimes 1_{Z})\circ \sigma_{Z,a}=\sigma_{Z,b}\circ (1_{Z}\otimes f)$, with graphical representation

\medskip

\begin{center}
\begin{minipage}[c]{0.17\textwidth}
\begin{tikzpicture}[scale=0.9]
  
  \draw[black, line width=1.2pt] (0,2) .. controls (0,1.2) and (1,0.8) .. (1,0);

  \draw[red, line width=1.2pt] (1,2) .. controls (1,1.2) and (0,0.8) .. (0,0);

  \node[draw, fill=white, inner sep=1.2pt, font=\scriptsize] at (0.12,1.5) {$f$};

  \node[below] at (0,0) {$Z$};
  \node[below] at (1,0) {$a$};
  \node[above] at (0,2) {$b$};
  \node[above] at (1,2) {$Z$};
\end{tikzpicture}
\end{minipage}
\begin{minipage}[c]{0.05\textwidth}
=
\end{minipage}
\begin{minipage}[c]{0.1\textwidth}
\begin{tikzpicture}[scale=0.9]

  \draw[black, line width=1.2pt] (0,2) .. controls (0,1.2) and (1,0.8) .. (1,0);

  \draw[red, line width=1.2pt] (1,2) .. controls (1,1.2) and (0,0.8) .. (0,0);

  \node[draw, fill=white, inner sep=1.2pt, font=\scriptsize] at (0.8,0.6) {$f$};

  \node[below] at (0,0) {$Z$};
  \node[below] at (1,0) {$a$};
  \node[above] at (0,2) {$b$};
  \node[above] at (1,2) {$Z$};
\end{tikzpicture}
\end{minipage}
\end{center}

\noindent Half-braidings are usually required to be normalized, i.e. $\sigma_{Z,\mathbbm{1}}=1_{Z}$, and are required to satisfy the following hexagon axiom (suppressing associators; or since we are just considering examples, assuming our categories are strict)

\begin{center}
\begin{minipage}[c]{0.22\textwidth}
\begin{tikzpicture}[scale=0.7]
  
  \draw[red, line width=1.2pt]
    (0,0) .. controls (0,0.8) and (2,1.2) .. (2,2);

  \draw[black, line width=1.2pt] (1,2) -- (1,0);

  \node[above] at (2,2) {$Z$};
  \node[above] at (1,2) {$a\otimes b$};

  \node[below] at (0,0) {$Z$};
  \node[below] at (1,0) {$a\otimes b$};
\end{tikzpicture}
\end{minipage}
\begin{minipage}[c]{0.06\textwidth}
\[
=
\]
\end{minipage}
\begin{minipage}[c]{0.28\textwidth}
\begin{tikzpicture}[scale=0.7]
  
  \draw[red, line width=1.2pt]
    (0,0) .. controls (0,0.8) and (2.4,1.2) .. (2.4,2);

  \draw[black, line width=1.2pt] (0.8,2) -- (0.8,0);
  \draw[black, line width=1.2pt] (1.6,2) -- (1.6,0);

  \node[above] at (2.4,2) {$Z$};
  \node[above] at (0.8,2) {$a$};
  \node[above] at (1.6,2) {$b$};

  \node[below] at (0,0) {$Z$};
  \node[below] at (0.8,-0.05) {$a$};
  \node[below] at (1.6,0) {$b$};
\end{tikzpicture}
\end{minipage}
\end{center}

\noindent The morphisms between objects $(Z,\sigma)$ and $(W, \gamma)$ are morphisms $f\in \mathcal{D}(Z,W)$ such that 

\begin{center}
\begin{minipage}[c]{0.17\textwidth}
\begin{tikzpicture}[scale=0.9]
  \draw[black, line width=1.2pt] (0,2) .. controls (0,1.2) and (1,0.8) .. (1,0);

  \draw[red, line width=1.2pt] (1,2) .. controls (1,1.2) and (0,0.8) .. (0,0);

  \node[draw, fill=white, inner sep=1.2pt, font=\scriptsize] at (0.84,1.5) {$f$};

  \node[below] at (0,0) {$Z$};
  \node[below] at (1,0) {$a$};
  \node[above] at (0,2) {$a$};
  \node[above] at (1,2) {$W$};
\end{tikzpicture}
\end{minipage}
\begin{minipage}[c]{0.05\textwidth}
=
\end{minipage}
\begin{minipage}[c]{0.1\textwidth}
\begin{tikzpicture}[scale=0.9]
  
  \draw[black, line width=1.2pt] (0,2) .. controls (0,1.2) and (1,0.8) .. (1,0);

  \draw[red, line width=1.2pt] (1,2) .. controls (1,1.2) and (0,0.8) .. (0,0);

  \node[draw, fill=white, inner sep=1.2pt, font=\scriptsize] at (0.22,0.6) {$f$};

  \node[below] at (0,0) {$Z$};
  \node[below] at (1,0) {$a$};
  \node[above] at (0,2) {$a$};
  \node[above] at (1,2) {$W$};
\end{tikzpicture}
\end{minipage}
\end{center}

\noindent The tensor product is defined as  $(Z,\sigma)\otimes (W,\gamma)=(Z\otimes W, \sigma\boxtimes \gamma)$, where $(\sigma\boxtimes \gamma)_{Z\otimes W, a}:=(\sigma_{Z,a}\otimes 1_{W})\circ(1_{Z}\otimes \gamma_{W,a}),$

\begin{center}
\begin{tikzpicture}[scale=1]
\node at (-1.8,1) {$(\sigma\boxtimes \gamma)_{Z\otimes W, a}=$};
  
  \draw[red, line width=1.2pt]
    (0.2,0) .. controls (0.5,0.8) and (1.3,1.2) .. (1.6,2);

  \draw[red, line width=1.2pt]
    (0.6,0) .. controls (0.9,0.8) and (1.7,1.2) .. (2.0,2);

  \draw[black, line width=1.2pt]
    (2.1,0) .. controls (1.8,0.8) and (0.8,1.2) .. (0.5,2);

  \node[above] at (1.6,2) {$Z$};
  \node[above] at (2.0,2) {$W$};
  \node[above] at (0.5,2) {$a$};

  \node[below] at (0.2,0) {$Z$};
  \node[below] at (0.6,0) {$W$};
  \node[below] at (2.1,0) {$a$};
\end{tikzpicture}
\end{center}

\noindent A standard result \cite{Izumi2000LongoRehrenI, MR1966525} is that $\mathcal{Z}(\mathcal{D})$ is a unitary braided fusion category, with braiding $(Z,\sigma)\otimes (W, \gamma)\cong (W,\gamma)\otimes (Z,\sigma)$ given simply by the half-braiding $\sigma_{Z,W}$.

The Drinfeld center $\mathcal{Z}(\mathcal{D})$ is very important in the general theory of fusion categories since it controls Morita  equivalence. In particular, two (unitary) fusion categories are Morita equivalent if and only if their Drinfeld centers are equivalent as unitary braided fusion categories.

\begin{thm}\cite{jones2024dhr}
$\text{DHR}(A(\mathcal{D},X))\cong \mathcal{Z}(\mathcal{D})$ for any fusion spin chain.
\end{thm}

We note that the original version of the theorem assumed $\mathcal{D}$ was fusion rather than indecomposable multi-fusion, and that $X$ was both self-dual and strong tensor generating rather than just strong tensor generating. These assumptions were made to simplify the argument, so that classical subfactor theory could be used ``out of the box". However, the result of \cite{hataishi2025structuredhrbimodulesabstract} shows that the self-dual assumption is not necessary, and the argument in \cite[Appendix A]{jones2025holographybulkboundarylocaltopological} that the assumption of fusion (rather than indecomposable multi-fusion) is not necessary. We will here explicitly sketch the functor $F_{\mathbbm{Z}}:\mathcal{Z}(\mathcal{D})\rightarrow \text{DHR}(A(\mathcal{D},X))$ from the original proof of the above theorem, which is the same construction but illustrated with pictures.
\medskip

\bigskip

For any interval $I$ with $|I|\ge n$, we have a natural tensor functor $F_{I}: \mathcal{Z}(\mathcal{D})\rightarrow \text{Corr}(A_{I})$ given by 

$$F_{I}(Z,\sigma):=\mathcal{D}(X^{\otimes I}, X^{\otimes I}\otimes Z)$$ which can be graphically represented as

$$\begin{tikzpicture}[scale=0.7]

  \node at (-3.1,0) {$\xi\in F_I(Z,\sigma)\;\leftrightsquigarrow$};

  \draw[thick] (-0.6,0.6) rectangle (0.6,-0.6);
  \node at (0,0) {$\xi$};

  \draw[thick] (-0.45,1.4) -- (-0.45,0.6);
  \draw[thick] (-0.15,1.4) -- (-0.15,0.6);
  \draw[thick] (0.15,1.4) -- (0.15,0.6);
  \draw[thick] (0.45,1.4) -- (0.45,0.6);

  \draw[thick] (-0.45,-0.6) -- (-0.45,-1.4);
  \draw[thick] (-0.15,-0.6) -- (-0.15,-1.4);
  \draw[thick] (0.15,-0.6) -- (0.15,-1.4);
  \draw[thick] (0.45,-0.6) -- (0.45,-1.4);

  \draw[red, thick] (0.6,0.6) -- (1.05,1.4);
  \node[above] at (1.05,1.4) {$Z$};

  \draw[decorate,decoration={brace,mirror,amplitude=5pt}]
    (-0.55,-1.55) -- (0.55,-1.55)
    node[midway,below=6pt] {$I$};

\end{tikzpicture}
$$

\noindent The left and right $A_{I}$ actions are given by $f\triangleright \xi\triangleleft g:=(f\otimes 1_{Z})\circ \xi\circ g$, while the right $A_{I}$-valued inner product is given by $\langle \xi \mid \nu\rangle := \xi^{\dagger}\circ \nu$. These can be represented graphically as 

\medskip

\begin{tikzpicture}[scale=0.7]
\begin{scope}
  
  \node at (-3.1,0) {$f\triangleright \xi\triangleleft g\;=$};

  \draw[thick] (-0.6,1.8) rectangle (0.6,0.9);
  \node at (0,1.35) {$f$};

  \draw[thick] (-0.6,0.6) rectangle (0.6,-0.6);
  \node at (0,0) {$\xi$};

  \draw[thick] (-0.6,-0.9) rectangle (0.6,-1.8);
  \node at (0,-1.35) {$g$};

  \draw[thick] (-0.45,2.6) -- (-0.45,1.8);
  \draw[thick] (-0.15,2.6) -- (-0.15,1.8);
  \draw[thick] (0.15,2.6) -- (0.15,1.8);
  \draw[thick] (0.45,2.6) -- (0.45,1.8);

  \draw[thick] (-0.45,0.9) -- (-0.45,0.6);
  \draw[thick] (-0.15,0.9) -- (-0.15,0.6);
  \draw[thick] (0.15,0.9) -- (0.15,0.6);
  \draw[thick] (0.45,0.9) -- (0.45,0.6);

  \draw[thick] (-0.45,-0.6) -- (-0.45,-0.9);
  \draw[thick] (-0.15,-0.6) -- (-0.15,-0.9);
  \draw[thick] (0.15,-0.6) -- (0.15,-0.9);
  \draw[thick] (0.45,-0.6) -- (0.45,-0.9);

  \draw[thick] (-0.45,-1.8) -- (-0.45,-2.6);
  \draw[thick] (-0.15,-1.8) -- (-0.15,-2.6);
  \draw[thick] (0.15,-1.8) -- (0.15,-2.6);
  \draw[thick] (0.45,-1.8) -- (0.45,-2.6);

  \draw[red, thick] (0.6,0.6) .. controls (1.15,1.15) and (1.05,1.95) .. (1.05,2.6);
  \node[above] at (1.05,2.6) {$Z$};

  \draw[decorate,decoration={brace,mirror,amplitude=5pt}]
    (-0.55,-2.75) -- (0.55,-2.75)
    node[midway,below=6pt] {$I$};
\end{scope}

 \node at (3,0) {and};

\begin{scope}[xshift=9cm]
\node at (-2.1,0) {$\langle \nu\ |\ \xi\rangle=$};
  
  \draw[thick] (-0.6,1.8) rectangle (0.6,0.6);
  \node at (0,1.2) {$\nu^{\dagger}$};

  \draw[thick] (-0.6,0.3) rectangle (0.6,-0.9);
  \node at (0,-0.3) {$\xi$};

  \draw[thick] (-0.45,2.6) -- (-0.45,1.8);
  \draw[thick] (-0.15,2.6) -- (-0.15,1.8);
  \draw[thick] (0.15,2.6) -- (0.15,1.8);
  \draw[thick] (0.45,2.6) -- (0.45,1.8);

  \draw[thick] (-0.45,0.6) -- (-0.45,0.3);
  \draw[thick] (-0.15,0.6) -- (-0.15,0.3);
  \draw[thick] (0.15,0.6) -- (0.15,0.3);
  \draw[thick] (0.45,0.6) -- (0.45,0.3);

  \draw[thick] (-0.45,-0.9) -- (-0.45,-1.7);
  \draw[thick] (-0.15,-0.9) -- (-0.15,-1.7);
  \draw[thick] (0.15,-0.9) -- (0.15,-1.7);
  \draw[thick] (0.45,-0.9) -- (0.45,-1.7);

  \draw[red, thick] (0.6,0.3) .. controls (0.78,0.45) and (0.78,0.45) .. (0.6,0.6);

  \draw[decorate,decoration={brace,mirror,amplitude=5pt}]
    (-0.55,-1.85) -- (0.55,-1.85)
    node[midway,below=6pt] {$I$};
\end{scope}
\end{tikzpicture}

\noindent

This assignment is functorial: if $f\in \mathcal{Z}(\mathcal{D})((Z,\sigma),(W,\gamma))$, then $F_{I}(f)(\xi):=(1_{X^{\otimes I}}\otimes f)\circ \xi$. Since $X^{\otimes n}$ contains all the simple objects, this naturally extends to a tensor functor, with tensorator $\mu^{I}_{(Z,\sigma),(W,\gamma)}: F_{I}(Z,\sigma)\otimes_{A_{I}} F_{I}(W,\gamma)\cong F_{I}(Z\otimes W, \sigma\boxtimes \gamma)$ given by composition $\xi\otimes \eta\mapsto (\xi\otimes 1_{F(b)})\circ \eta\in F_{I}(Z\otimes W, \sigma\boxtimes \gamma)$ \cite{chen2022ktheoretic}. One can easily check that this extends to an isometric map of correspondences from the definitions. To check that this is an actual unitary isomorphism, we use the assumption that $X$ is strong tensor generating. This implies that there is a collection $\{e^{Z}_{i}\}\subseteq F_{I}(Z,\sigma)$ satisfying $\sum_{i} e^{Z}_{i}\circ (e^{Z}_{i})^{\dagger}=1_{X^{\otimes n}}\otimes 1_{Z}$.

$$\begin{tikzpicture}[scale=0.7, baseline={(current bounding box.center)}]
  
  \begin{scope}[xshift=0cm]
    \draw[thick] (0,0) -- (0,4.8);
    \draw[thick] (0.8,0) -- (0.8,4.8);
    \draw[thick] (1.6,0) -- (1.6,4.8);
    \draw[thick] (2.4,0) -- (2.4,4.8);
    \draw[thick,red] (3.4,0) -- (3.4,4.8);
  \end{scope}

  \node at (5.4,2) {$\displaystyle =\sum_i$};

  \begin{scope}[xshift=7.2cm]
    
    \draw[thick] (0,0) -- (0,1.2);
    \draw[thick] (0.8,0) -- (0.8,1.2);
    \draw[thick] (1.6,0) -- (1.6,1.2);
    \draw[thick] (2.4,0) -- (2.4,1.2);
    \draw[thick,red] (3.4,0) -- (3.4,1.2);

    \draw[fill=white] (-0.3,1.2) rectangle (3.7,2.0);
    \node at (1.7,1.6) {$(e^{Z}_i)^{\dagger}$};

    \draw[thick] (0,2.0) -- (0,2.8);
    \draw[thick] (0.8,2.0) -- (0.8,2.8);
    \draw[thick] (1.6,2.0) -- (1.6,2.8);
    \draw[thick] (2.4,2.0) -- (2.4,2.8);

    \draw[fill=white] (-0.3,2.8) rectangle (3.7,3.6);
    \node at (1.7,3.2) {$e^{Z}_{i}$};

    \draw[thick] (0,3.6) -- (0,4.8);
    \draw[thick] (0.8,3.6) -- (0.8,4.8);
    \draw[thick] (1.6,3.6) -- (1.6,4.8);
    \draw[thick] (2.4,3.6) -- (2.4,4.8);
    \draw[thick,red] (3.4,3.6) -- (3.4,4.8);
  \end{scope}
\end{tikzpicture}
$$

\noindent Then this gives a projective basis (a.k.a. Pimsner-Popa/Watatani basis) for the right Hilbert module $F_{I}(Z,\sigma)$. Up to this point, we haven't actually used the half-braiding $\sigma$. We will us this to define inclusions of correspondences and then take the inductive limit.

Now, if $I\subseteq J$, we have a natural inclusion $j_{I\subseteq J}: F_{I}(Z,\sigma)\hookrightarrow F_{J}(Z,\sigma)$ defined by 

$$\xi\mapsto 1_{X^{\otimes (J<I)}}\otimes \left(( 1_{X^{\otimes I}}\otimes \sigma_{Z,X^{\otimes J>I}})\circ (\xi \otimes 1_{X^{\otimes (J>I)}})\right),
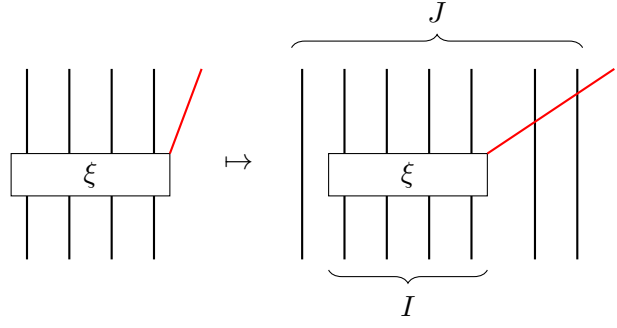
$$

\medskip

\bigskip

\begin{center}
\begin{tikzpicture}[scale=0.7, baseline={(current bounding box.center)}]

  \begin{scope}[xshift=0cm]

    \draw[thick] (0,0) -- (0,1.2);
    \draw[thick] (0.8,0) -- (0.8,1.2);
    \draw[thick] (1.6,0) -- (1.6,1.2);
    \draw[thick] (2.4,0) -- (2.4,1.2);

    \draw[fill=white] (-0.3,1.2) rectangle (2.7,2.0);
    \node at (1.2,1.6) {$\xi$};

    \draw[thick] (0,2.0) -- (0,3.6);
    \draw[thick] (0.8,2.0) -- (0.8,3.6);
    \draw[thick] (1.6,2.0) -- (1.6,3.6);
    \draw[thick] (2.4,2.0) -- (2.4,3.6);

    \draw[thick,red] (2.7,2.0) -- (3.3,3.6);
  \end{scope}

  \node at (4,1.8) {$\mapsto$};

  \begin{scope}[xshift=6cm]
   
    \draw[thick] (-0.8,0) -- (-0.8,3.6);
    \draw[thick] (0,0) -- (0,1.2);
    \draw[thick] (0.8,0) -- (0.8,1.2);
    \draw[thick] (1.6,0) -- (1.6,1.2);
    \draw[thick] (2.4,0) -- (2.4,1.2);

    \draw[thick] (3.6,0) -- (3.6,3.6);
    \draw[thick] (4.4,0) -- (4.4,3.6);

    \draw[fill=white] (-0.3,1.2) rectangle (2.7,2.0);
    \node at (1.2,1.6) {$\xi$};

    \draw[thick] (0,2.0) -- (0,3.6);
    \draw[thick] (0.8,2.0) -- (0.8,3.6);
    \draw[thick] (1.6,2.0) -- (1.6,3.6);
    \draw[thick] (2.4,2.0) -- (2.4,3.6);

    \draw[thick,red] (2.7,2.0) -- (5.1,3.6);
       \draw[decorate,decoration={brace,mirror,amplitude=5pt}]
    (-0.3,-0.2) -- (2.7,-0.2)
    node[midway,below=6pt] {$I$};
     \draw[decorate,decoration={brace,amplitude=5pt}]
    (-1,4) -- (4.5,4)
    node[midway,above=6pt] {$J$};
  \end{scope}

\end{tikzpicture}
\end{center}

\bigskip

\noindent \noindent By the properties of the half-braiding, this extends to an equivariant inclusion of actions $F_{I}\rightarrow F_{J}$, and thus extends to a tensor functor $F_{\mathbbm{Z}}: \mathcal{C}\rightarrow \text{Corr}(A)$. Now, for any sufficiently large interval $I$, if we take the projective bases built for $F_{I}(Z,\sigma)$, this extends to a projective basis for $F_{\mathbbm{Z}}(Z,\sigma)$. But by naturality of the half-braiding, this is localized in $I$:

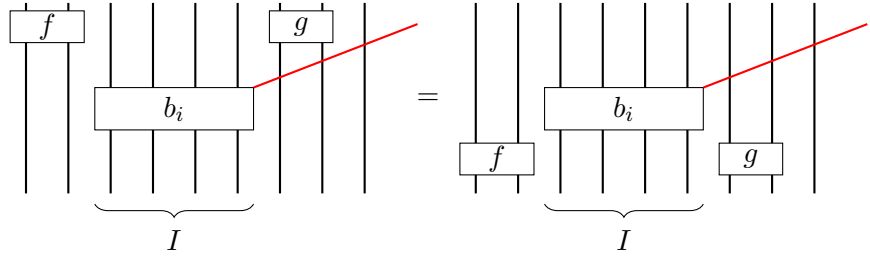
\bigskip

\begin{center}

\begin{tikzpicture}[scale=0.7]
\begin{scope}[xshift=0cm]
   
    \draw[thick] (-1.6,0) -- (-1.6,3.6);
    \draw[thick] (-0.8,0) -- (-0.8,3.6);
    \draw[thick] (0,0) -- (0,1.2);
    \draw[thick] (0.8,0) -- (0.8,1.2);
    \draw[thick] (1.6,0) -- (1.6,1.2);
    \draw[thick] (2.4,0) -- (2.4,1.2);

    \draw[thick] (3.2,0) -- (3.2,3.6);
    \draw[thick] (4.0,0) -- (4.0,3.6);
    \draw[thick] (4.8,0) -- (4.8,3.6);

    \draw[fill=white] (-0.3,1.2) rectangle (2.7,2.0);
    \node at (1.2,1.6) {$b_i$};

    \draw[thick] (0,2.0) -- (0,3.6);
    \draw[thick] (0.8,2.0) -- (0.8,3.6);
    \draw[thick] (1.6,2.0) -- (1.6,3.6);
    \draw[thick] (2.4,2.0) -- (2.4,3.6);

    \draw[thick,red] (2.7,2.0) -- (5.8,3.2);

    \draw[fill=white] (-1.9,2.85) rectangle (-0.5,3.45);
    \node at (-1.2,3.15) {$f$};

    \draw[fill=white] (3.0,2.85) rectangle (4.2,3.45);
    \node at (3.6,3.15) {$g$};
      \draw[decorate,decoration={brace,mirror,amplitude=5pt}]
    (-0.3,-0.2) -- (2.7,-0.2)
    node[midway,below=6pt] {$I$};
\end{scope}

\node at (6.0,1.8) {$=$};

\begin{scope}[xshift=8.5cm]
    
    \draw[thick] (-1.6,0) -- (-1.6,3.6);
    \draw[thick] (-0.8,0) -- (-0.8,3.6);
    \draw[thick] (0,0) -- (0,1.2);
    \draw[thick] (0.8,0) -- (0.8,1.2);
    \draw[thick] (1.6,0) -- (1.6,1.2);
    \draw[thick] (2.4,0) -- (2.4,1.2);

    \draw[thick] (3.2,0) -- (3.2,3.6);
    \draw[thick] (4.0,0) -- (4.0,3.6);
    \draw[thick] (4.8,0) -- (4.8,3.6);

    \draw[fill=white] (-0.3,1.2) rectangle (2.7,2.0);
    \node at (1.2,1.6) {$b_i$};

    \draw[thick] (0,2.0) -- (0,3.6);
    \draw[thick] (0.8,2.0) -- (0.8,3.6);
    \draw[thick] (1.6,2.0) -- (1.6,3.6);
    \draw[thick] (2.4,2.0) -- (2.4,3.6);

    \draw[thick,red] (2.7,2.0) -- (5.8,3.2);

    \draw[fill=white] (-1.9,0.35) rectangle (-0.5,0.95);
    \node at (-1.2,0.65) {$f$};

    \draw[fill=white] (3.0,0.35) rectangle (4.2,0.95);
    \node at (3.6,0.65) {$g$};
      \draw[decorate,decoration={brace,mirror,amplitude=5pt}]
    (-0.3,-0.2) -- (2.7,-0.2)
    node[midway,below=6pt] {$I$};
\end{scope}

\end{tikzpicture}

\end{center}

\noindent Thus $F_{\mathbbm{Z}}:\mathcal{Z}(\mathcal{D})\rightarrow \text{DHR}(A)$ extends to a monoidal functor. It turns out that $F_{\mathbbm{Z}}$ is a braided equivalence, but proving this requires significantly more work. Nevertheless, it is useful to have this picture in mind.

\subsection{The standard model: anyon chains}\label{sec:anyonchains}
We can now describe the standard model for fusion categorical symmetries on the lattice: anyon chains. We present this in a way most convenient for our subfactor-theoretic perspective, but there are equivalent formulations in terms of matrix product operators and/or weak Hopf algebras.

Let $\mathcal{C}$ be a unitary fusion category. To build an example of a quasi-local algebra, the input for this construction is a right indecomposable $\mathcal{C}$-module category $\mathcal{M}$ \cite{MR1976459}. Let $\mathcal{D}:=\mathcal{C}^{*}_{\mathcal{M}}= \text{End}_{\mathcal{C}}(\mathcal{M})$ be the dual category, which in this context can be thought of as the ``charge" category. We also choose a strong tensor generator $X\in \mathcal{D}$.

Now, let $\mathcal{E}:=\text{End}(\mathcal{M})\cong \text{Mat}_{n}(\text{Hilb}_{f.d.})$, where $n=\text{rank}(\mathcal{M})$. Then $\mathcal{E}$ is an indecomposable, Morita trivial multi-fusion category.
$\mathcal{M}$ is naturally equipped with the structure of an invertible $\mathcal{D}$-$\mathcal{C}$ bimodule category.

There is a natural forgetful inclusion $\mathcal{D}=\text{End}_{\mathcal{C}}(\mathcal{M})\rightarrow \text{End}(\mathcal{M})=\mathcal{E}$, and thus we have a natural inclusion 

$$A(\mathcal{D},X)_{I}\hookrightarrow A(\mathcal{E},X)_{I},$$

\noindent where in the latter, we identify $X$ with the underlying functor on $\mathcal{M}$ (forgetting the module structure). Identifying $A(\mathcal{D},X)_{I}$ with its image extends to an inclusion $A(\mathcal{D},X)\subseteq A(\mathcal{E},X)$. We recognize the latter quasi-local algebra as an anyon chain. We have the following from \cite{jones2025quantumcellularautomatacategorical, evans2026operatoralgebraicapproachfusion}:

\begin{thm}
For any indecomposable $\mathcal{C}$-module category $\mathcal{M}$ and strong tensor generator $X\in \mathcal{C}^{*}_{\mathcal{M}}$, $A(\mathcal{D},X)\hookrightarrow A(\mathcal{E},X)$ is a physical boundary subalgebra, with symmetry category $\mathcal{C}$.
\end{thm}

In \cite{jones2025quantumcellularautomatacategorical}, these inclusions were called \textit{spatial realizations}, a nod to the graph-based Hilbert space structure underlying an anyon chain. Here, we will call them \textit{anyon models of the fusion category symmetry}. Recall that there is a canonical correspondence between \textit{left} $\mathcal{D}$ module categories and Lagrangian algebras in $\mathcal{Z}(\mathcal{D})$. Given our setup, if we view $\mathcal{M}$ as a left $\mathcal{D}$-module category then we denote the associated Lagrangian algebra (normalized to be a dual Q-system) by $L$. Then $A(\mathcal{E},X)$, as a (dual) Q-system in $\text{DHR}(A(\mathcal{D},X))\cong \mathcal{Z}(\mathcal{D})$, is precisely $L$.

Although we know abstractly that $\mathcal{C}\cong \text{DHR}_{-}(A(\mathcal{E},X), A(\mathcal{D},X))$, we want to write this equivalence down explicitly. Namely, how do we see $\mathcal{C}$ as correspondences of $A(\mathcal{E},X)$?

We have a natural unitary tensor functor from the monoidal opposite $R:\mathcal{C}^{mp}\rightarrow \mathcal{E}=\text{End}(\mathcal{M})$, given from the right module structure by $R_{a}(m):=m\triangleleft a$, and tensorator given by the natural transformation

$$\mu_{a,b}:R_{a}\circ R_{b}(m)=(m\triangleleft b)\triangleleft a\cong m\triangleleft (b\otimes a)=R_{b\otimes a}(m)=R_{a\otimes^{mp}b}(m).$$

In fact, $R$ gives an equivalence $$R:\mathcal{C}^{mp}\cong \mathcal{Z}_{\mathcal{E}}(\mathcal{D}),$$

\noindent where $\mathcal{Z}_{\mathcal{E}}(\mathcal{D})$ is the \textit{centralizer} of $\mathcal{D}$ in the larger multi-fusion category $\mathcal{E}$. Here $\mathcal{Z}_{\mathcal{E}}(\mathcal{D})$ consists of pairs $(E,\sigma)$, where $E\in \mathcal{E}$ and $\sigma=\{\sigma_{E,Y}: Y\otimes E\cong E\otimes Y\}_{Y\in \mathcal{D}}$ is a unitary half-braiding on $Y$ but only with respect to objects in $\mathcal{D}$. Furthermore, it should only be natural with respect to morphisms in $\mathcal{D}$ rather than all morphisms in $\mathcal{E}$, since the inclusion of categories $\mathcal{D}\subseteq \mathcal{E}$ is not full (or in other words, there are typically many more morphisms in $\mathcal{E}$ than in $\mathcal{D}$). These are required to satisfy the same type of conditions (e.g. normalization, hexagon axiom) as half-braidings for the Drinfeld center.

\begin{center}
\begin{minipage}[c]{0.1cm}
    $$\sigma_{E,Y}:=$$
\end{minipage}
\begin{minipage}[c]{0.3\textwidth}
    $$
\begin{tikzpicture}[scale=1.2]
  
  \draw[red, line width=1.2pt] (0,2) .. controls (0,1.2) and (1,0.8) .. (1,0);

  \draw[black, line width=1.2pt] (1,2) .. controls (1,1.2) and (0,0.8) .. (0,0);

  \node[below] at (0,0) {$Y$};
  \node[below] at (1,0) {$E$};
\end{tikzpicture}
$$
\end{minipage}
\end{center}

As mentioned above, the monoidal functor $R:\mathcal{C}^{mp}\rightarrow \text{End}(\mathcal{M})=\mathcal{E}$ extends to an equivalence $R:\mathcal{C}^{mp}\cong \mathcal{Z}_{\mathcal{E}}(\mathcal{D})$, by naturally equipping the objects $R(a)$ with the $\mathcal{D}$ half-braidings

$$(Y\otimes R(a))(m)=Y(R(a)(m))=Y(m\triangleleft a)\cong Y(m)\triangleleft a=R(a)(Y(m))=(R(a)\otimes Y)(m),$$

\noindent where the middle isomorphism uses the $\mathcal{C}$-module functor structure on $Y$, and naturality follows from the naturality of the module functor structure.

Now, using this, we can build a model for bimodules in $\text{DHR}_{-}(A(\mathcal{E},X)| A(\mathcal{D},X))$ in a manner very similar to how we constructed DHR bimodules for $A(\mathcal{D},X)$ in the previous section. The only differences are that 

\begin{enumerate}
    \item 
    We use morphisms in $\mathcal{E}$ rather than morphisms in $\mathcal{D}$;
    \item 
    We have our red strings (labelled by $R(a)$ and indexing the bimodule) pulled to the \textit{left}, rather than right. 
\end{enumerate}

\noindent Otherwise, the construction proceeds essentially verbatim. We define, for each $I$, $G(a)_{I}:=\text{Hom}_{\mathcal{E}}(X^{\otimes I}, R(a)\otimes X^{\otimes I})$. Graphically, we have the representation

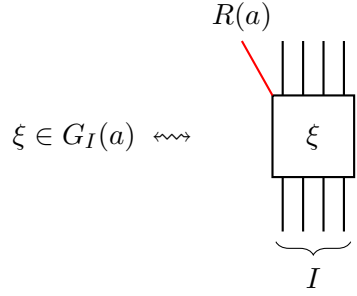
\begin{center}
\begin{tikzpicture}[scale=0.9]

  \node at (-3.1,0) {$\xi\in G_I(a)\;\leftrightsquigarrow$};

  \draw[thick] (-0.6,0.6) rectangle (0.6,-0.6);
  \node at (0,0) {$\xi$};

  \draw[thick] (-0.45,1.4) -- (-0.45,0.6);
  \draw[thick] (-0.15,1.4) -- (-0.15,0.6);
  \draw[thick] (0.15,1.4) -- (0.15,0.6);
  \draw[thick] (0.45,1.4) -- (0.45,0.6);

  \draw[thick] (-0.45,-0.6) -- (-0.45,-1.4);
  \draw[thick] (-0.15,-0.6) -- (-0.15,-1.4);
  \draw[thick] (0.15,-0.6) -- (0.15,-1.4);
  \draw[thick] (0.45,-0.6) -- (0.45,-1.4);

  \draw[red, thick] (-0.6,0.6) -- (-1.05,1.4);
  \node[above] at (-1.05,1.4) {$R(a)$};

  \draw[decorate,decoration={brace,mirror,amplitude=5pt}]
    (-0.55,-1.55) -- (0.55,-1.55)
    node[midway,below=6pt] {$I$};

\end{tikzpicture}
\end{center}

We have verbatim the same $A_{I}$-bimodule and inner product structure as in the Drinfeld center case. The first difference we see is that since our strings are to the left, stacking pictures \textit{reverses} the order of tensor products. But since $R$ also reverses tensor products, $G$ orients the tensor products correctly: 

$$G_{I}(a)\boxtimes_{A_{I}} G_{I}(b)\cong G_{I}(a\otimes b).$$

Our inclusion also works similarly. Namely if $I\subseteq J$, we have the inclusion

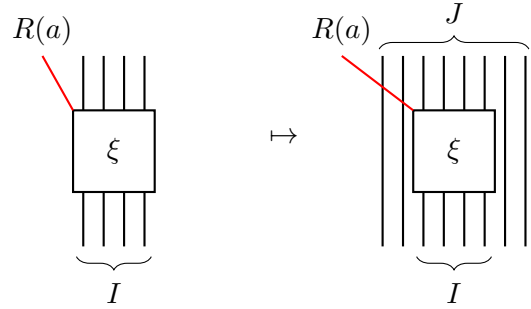
\begin{center}
\begin{tikzpicture}[scale=0.9]

\begin{scope}
  
  \draw[thick] (-0.6,0.6) rectangle (0.6,-0.6);
  \node at (0,0) {$\xi$};

  \draw[thick] (-0.45,1.4) -- (-0.45,0.6);
  \draw[thick] (-0.15,1.4) -- (-0.15,0.6);
  \draw[thick] (0.15,1.4) -- (0.15,0.6);
  \draw[thick] (0.45,1.4) -- (0.45,0.6);

  \draw[thick] (-0.45,-0.6) -- (-0.45,-1.4);
  \draw[thick] (-0.15,-0.6) -- (-0.15,-1.4);
  \draw[thick] (0.15,-0.6) -- (0.15,-1.4);
  \draw[thick] (0.45,-0.6) -- (0.45,-1.4);

  \draw[red, thick] (-0.6,0.6) -- (-1.05,1.4);
  \node[above] at (-1.05,1.4) {$R(a)$};

  \draw[decorate,decoration={brace,mirror,amplitude=5pt}]
    (-0.55,-1.55) -- (0.55,-1.55)
    node[midway,below=6pt] {$I$};

\end{scope}

  \node at (2.5,0.2) {$\mapsto$};

\begin{scope}[xshift=5cm]
  
  \draw[thick] (-0.6,0.6) rectangle (0.6,-0.6);
  \node at (0,0) {$\xi$};

  \draw[thick] (-0.45,1.4) -- (-0.45,0.6);
  \draw[thick] (-0.15,1.4) -- (-0.15,0.6);
  \draw[thick] (0.15,1.4) -- (0.15,0.6);
  \draw[thick] (0.45,1.4) -- (0.45,0.6);

  \draw[thick] (-0.45,-0.6) -- (-0.45,-1.4);
  \draw[thick] (-0.15,-0.6) -- (-0.15,-1.4);
  \draw[thick] (0.15,-0.6) -- (0.15,-1.4);
  \draw[thick] (0.45,-0.6) -- (0.45,-1.4);

  \draw[thick] (-0.75,1.4) -- (-0.75,-1.4);
  \draw[thick] (-1.05,1.4) -- (-1.05,-1.4);
  \draw[thick] (.75,1.4) -- (.75,-1.4);
  \draw[thick] (1.05,1.4) -- (1.05,-1.4);

  \draw[red, thick] (-0.6,0.6) -- (-1.65,1.4);
  \node[above] at (-1.65,1.4) {$R(a)$};

  \draw[decorate,decoration={brace,mirror,amplitude=5pt}]
    (-0.55,-1.55) -- (0.55,-1.55)
    node[midway,below=6pt] {$I$};
     \draw[decorate,decoration={brace,amplitude=5pt}]
    (-1.1,1.5) -- (1.1,1.5)
    node[midway,above=6pt] {$J$};

    \end{scope}

\end{tikzpicture}
\end{center}

This establishes an $A-A$ correspondence $G_{\mathbbm{Z}}(a)$ by taking the colimit as above, which extends to a (fully faithful) monoidal functor $G_{\mathbbm{Z}}:\mathcal{C}\rightarrow \text{Corr}(A)$. If we take one of the basis elements $b_{i}$ in $G_{I}(a)$ and look at its image in $G_{\mathbbm{Z}}(a)$, we see that for all $g\in A(\mathcal{E},X)_{>I}$, $gb_{i}=b_{i}g$,

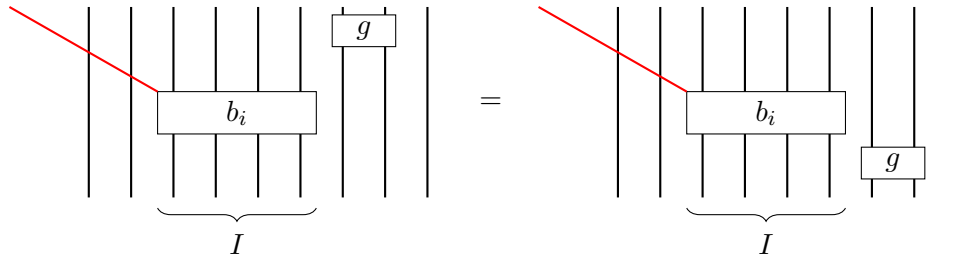
\bigskip

\begin{center}

\begin{tikzpicture}[scale=0.7]
\begin{scope}[xshift=0cm]
    
    \draw[thick] (-1.6,0) -- (-1.6,3.6);
    \draw[thick] (-0.8,0) -- (-0.8,3.6);
    \draw[thick] (0,0) -- (0,1.2);
    \draw[thick] (0.8,0) -- (0.8,1.2);
    \draw[thick] (1.6,0) -- (1.6,1.2);
    \draw[thick] (2.4,0) -- (2.4,1.2);

    \draw[thick] (3.2,0) -- (3.2,3.6);
    \draw[thick] (4.0,0) -- (4.0,3.6);
    \draw[thick] (4.8,0) -- (4.8,3.6);

    \draw[fill=white] (-0.3,1.2) rectangle (2.7,2.0);
    \node at (1.2,1.6) {$b_i$};

    \draw[thick] (0,2.0) -- (0,3.6);
    \draw[thick] (0.8,2.0) -- (0.8,3.6);
    \draw[thick] (1.6,2.0) -- (1.6,3.6);
    \draw[thick] (2.4,2.0) -- (2.4,3.6);

    \draw[thick,red] (-0.3,2.0) -- (-3.1,3.6);

    \draw[fill=white] (3.0,2.85) rectangle (4.2,3.45);
    \node at (3.6,3.15) {$g$};
      \draw[decorate,decoration={brace,mirror,amplitude=5pt}]
    (-0.3,-0.2) -- (2.7,-0.2)
    node[midway,below=6pt] {$I$};
\end{scope}

\node at (6.0,1.8) {$=$};

\begin{scope}[xshift=10cm]
    
    \draw[thick] (-1.6,0) -- (-1.6,3.6);
    \draw[thick] (-0.8,0) -- (-0.8,3.6);
    \draw[thick] (0,0) -- (0,1.2);
    \draw[thick] (0.8,0) -- (0.8,1.2);
    \draw[thick] (1.6,0) -- (1.6,1.2);
    \draw[thick] (2.4,0) -- (2.4,1.2);

    \draw[thick] (3.2,0) -- (3.2,3.6);
    \draw[thick] (4.0,0) -- (4.0,3.6);
    \draw[thick] (4.8,0) -- (4.8,3.6);

    \draw[fill=white] (-0.3,1.2) rectangle (2.7,2.0);
    \node at (1.2,1.6) {$b_i$};

    \draw[thick] (0,2.0) -- (0,3.6);
    \draw[thick] (0.8,2.0) -- (0.8,3.6);
    \draw[thick] (1.6,2.0) -- (1.6,3.6);
    \draw[thick] (2.4,2.0) -- (2.4,3.6);

    \draw[fill=white] (3.0,0.35) rectangle (4.2,0.95);
    \node at (3.6,0.65) {$g$};
      \draw[decorate,decoration={brace,mirror,amplitude=5pt}]
    (-0.3,-0.2) -- (2.7,-0.2)
    node[midway,below=6pt] {$I$};
     \draw[thick,red] (-0.3,2.0) -- (-3.1,3.6);
\end{scope}

\end{tikzpicture}

\end{center}

\bigskip
\noindent

\noindent However, since the half-braiding is only natural for morphisms in $\mathcal{D}$, we can only say in general for $f\in A(\mathcal{D},X)_{<I}$ that $fb_{i}=b_{i}f$,

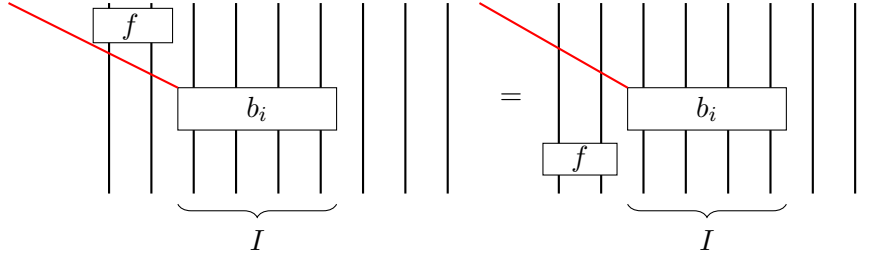
\bigskip

\begin{center}

\begin{tikzpicture}[scale=0.7]
\begin{scope}[xshift=0cm]
    \draw[thick] (-1.6,0) -- (-1.6,3.6);
    \draw[thick] (-0.8,0) -- (-0.8,3.6);
    \draw[thick] (0,0) -- (0,1.2);
    \draw[thick] (0.8,0) -- (0.8,1.2);
    \draw[thick] (1.6,0) -- (1.6,1.2);
    \draw[thick] (2.4,0) -- (2.4,1.2);

    \draw[thick] (3.2,0) -- (3.2,3.6);
    \draw[thick] (4.0,0) -- (4.0,3.6);
    \draw[thick] (4.8,0) -- (4.8,3.6);

    \draw[fill=white] (-0.3,1.2) rectangle (2.7,2.0);
    \node at (1.2,1.6) {$b_i$};

    \draw[thick] (0,2.0) -- (0,3.6);
    \draw[thick] (0.8,2.0) -- (0.8,3.6);
    \draw[thick] (1.6,2.0) -- (1.6,3.6);
    \draw[thick] (2.4,2.0) -- (2.4,3.6);

    \draw[thick,red] (-0.3,2.0) -- (-3.5,3.6);

    \draw[fill=white] (-1.9,2.85) rectangle (-0.4,3.5);
    \node at (-1.2,3.15) {$f$};

      \draw[decorate,decoration={brace,mirror,amplitude=5pt}]
    (-0.3,-0.2) -- (2.7,-0.2)
    node[midway,below=6pt] {$I$};
\end{scope}

\node at (6.0,1.8) {$=$};

\begin{scope}[xshift=8.5cm]
    \draw[thick] (-1.6,0) -- (-1.6,3.6);
    \draw[thick] (-0.8,0) -- (-0.8,3.6);
    \draw[thick] (0,0) -- (0,1.2);
    \draw[thick] (0.8,0) -- (0.8,1.2);
    \draw[thick] (1.6,0) -- (1.6,1.2);
    \draw[thick] (2.4,0) -- (2.4,1.2);

    \draw[thick] (3.2,0) -- (3.2,3.6);
    \draw[thick] (4.0,0) -- (4.0,3.6);
    \draw[thick] (4.8,0) -- (4.8,3.6);

    \draw[fill=white] (-0.3,1.2) rectangle (2.7,2.0);
    \node at (1.2,1.6) {$b_i$};

    \draw[thick] (0,2.0) -- (0,3.6);
    \draw[thick] (0.8,2.0) -- (0.8,3.6);
    \draw[thick] (1.6,2.0) -- (1.6,3.6);
    \draw[thick] (2.4,2.0) -- (2.4,3.6);

    \draw[thick,red] (-0.3,2.0) -- (-3.1,3.6);

    \draw[fill=white] (-1.9,0.35) rectangle (-0.5,0.95);
    \node at (-1.2,0.65) {$f$};

      \draw[decorate,decoration={brace,mirror,amplitude=5pt}]
    (-0.3,-0.2) -- (2.7,-0.2)
    node[midway,below=6pt] {$I$};
\end{scope}

\end{tikzpicture}

\end{center}

\bigskip

\noindent Thus $G_{\mathbbm{Z}}:\mathcal{C}\rightarrow \text{DHR}_{-}(A(\mathcal{E},X),A(\mathcal{D},X))$, and in fact this is an equivalence. We now have a concrete model for the symmetry category $\mathcal{C}$ as correspondences of $A$.

\section{Infinite stabilization}

We now define a version of spin chains with infinite-dimensional, separable local Hilbert spaces $\ell^{2}(\mathbbm{N})$. In some corners of the condensed matter world, we may prefer the separable Hilbert space $L^{2}(S^{1})$ and call this a ``rotor chain;'' however, we do not need to specify a Hamiltonian so the actual form of the separable Hilbert space is immaterial. Using $\ell^{2}(\mathbbm{N})$ is typically associated with a quantum harmonic oscillator.   

We define 
$$\mathcal{L}^{\infty}:=A(\text{Hilb}, \ell^{2}(\mathbbm{N})).$$ 

More concretely, for any interval $I\subseteq \mathbbm{Z}$, define the local Hilbert space $(\ell^{2}(\mathbbm{N}))^{\otimes I}$, the $|I|$-fold tensor power of the Hilbert space $\ell^{2}(\mathbbm{N})$, with tensor factors indexed by sites in $I$. The local algebra $\mathcal{L}^{\infty}_{I}:=\mathcal{B}(\ell^{2}(\mathbbm{N})^{\otimes I})\cong \mathcal{B}(\ell^{2}(\mathbbm{N}))\overline{\otimes}\dots \overline{\otimes} \mathcal{B}(\ell^{2}(\mathbbm{N}))$, where $\mathcal{B}(\ell^{2}(\mathbbm{N}))$ denotes the von Neumann algebra of all bounded operators on the Hilbert space $\ell^{2}(\mathbbm{N})$, and $\overline{\otimes}$ denotes the tensor product of von Neumann algebras. The quasi-local C*-algebra is

$$\mathcal{L}^{\infty}:=\text{colim}_{I} \mathcal{L}^{\infty}_{I}.$$ 

\begin{remark}
Physically, $\mathcal{L}^{\infty}$ is realized explicitly by free bosons \cite{MR1441540} in the following sense. Consider the single-particle Hilbert space $\ell^{2}(\mathbbm{Z})$ and the orthonormal basis $e_{k}=\delta_{k}$ for $k\in \mathbbm{Z}$. Applying bosonic second quantization, we obtain an algebra of creation and annihilation operators $b^{\dagger}_{k}, b_{k}$,  where $k\in \mathbbm{Z}$, satisfying

$$[b^{\dagger}_{k}, b^{\dagger}_{j}]=[b_{k},b_{j}]=0$$

$$[b_{j}, b^{\dagger}_{k}]=\delta_{j,k} 1$$

This factorizes into a tensor product of algebras $\widetilde{L}_{k}$ generated by $1, b_{k}, b^{\dagger}_{k}$ for $k\in \mathbbm{Z}$ (thinking of the integers as spatial lattice sites), and for any finite $F\subseteq \mathbbm{Z}$, we have the subalgebra $\widetilde{L}_{F}\cong \otimes_{k\in F} \widetilde{L}_{k}$. Each of these algebras integrates to a Weyl algebra with one-dimensional underlying complex Hilbert space. In the standard Fock space representation (or any representation such that the appropriate one-parameter groups are continuous) the W*-algebra of bounded operators generated by $\widetilde{L}_{F}$ yields a copy of $\mathcal{B}(\ell^{2}(\mathbbm{N}))$. In fact, we see that the local von Neumann algebra generated by $\mathcal{L}_{F}:=\text{W}^{*}(\widetilde{L}_{F})\cong \mathcal{B}(\ell^{2}(\mathbbm{N}))^{\overline{\otimes} F}$. Thus by completing von Neumann algebras locally, we precisely obtain the quasi-local algebra $\mathcal{L}^{\infty}$.
\end{remark}

We will now see that $\mathcal{L}^{\infty}$ satisfies a strong form of algebraic Haag duality, where $F$ can be arbitrary finite subsets of $\mathbbm{Z}$.

\begin{thm}
For any finite subset $F\subseteq \mathbb{Z}$, we have $\{x\in \mathcal{L}^{\infty}\ :\ [x,\mathcal{L}^{\infty}_{F^{c}}]=0\}=\mathcal{L}^{\infty}_{F}$.
\end{thm}

\begin{proof}
    Let $F$ be a finite subset and take $a \in \mathcal{L}^{\infty} \cap (\mathcal{L}^{\infty}_{F^{c}})' = Z_{\mathcal{L}^{\infty}}(\mathcal{L}^{\infty}_{F^{c}})$. Without loss of generality assume $\lVert a\rVert=1$.
    Since $\mathcal{L}^{\infty}_{F}$ is a norm-closed subalgebra of $\mathcal{L}^{\infty}$, we will show that $a$ can be approximated arbitrarily closely by elements of $\mathcal{L}^{\infty}_{F}$. 
    Since $\mathcal{L}^{\infty}$ is the inductive limit of the $\mathcal{L}^{\infty}_{I}$ for intervals $I$ containing $F$, for any $\epsilon > 0$ we can find an interval $J$ containing $F$ and $b \in \mathcal{L}^{\infty}_{J}$ such that $\lVert a - b \rVert < \epsilon$.  We can assume, without loss of generality, that $\lVert b \rVert \leq 1$ by rescaling.

    By expanding and rearranging the commutator as follows
    \begin{align*} 
    [b, y] &= by -yb \\
           &= (b-a+a)y - y(b-a+a) \\
           &= (b-a)y + ay - y(b-a) - ya \\
           &= (b-a)y - y(b-a) + ay - ya \\
           &= [b-a, y] + [a,y] 
    \end{align*}
    we then compute that for any $y \in \mathcal{L}^{\infty}_{F^{c}}$
    \begin{align*}
        \left\lVert [b,y] \right\rVert &= \lVert [b-a,y] + [a,y] \rVert \\
        &\leq \lVert[b-a, y] \rVert + \lVert[a,y]\rVert \\
        &= \lVert [b-a, y] \rVert \\
        &= \lVert (b-a)y + y(a-b) \rVert \\
        &\leq \lVert b- a \rVert \lVert y \rVert + \lVert y \rVert \lVert a - b\rVert \\
        &\leq 2\epsilon \lVert y \rVert
    \end{align*}

    Since we assume $J \supseteq F$ we can write $\mathcal{L}^{\infty}_{J}$ as $\mathcal{B}((\ell^{2}(\mathbbm{N}))^{\bar{\otimes} J-I}) \bar{\otimes} \mathcal{B}((\ell^{2}(\mathbbm{N}))^{\bar{\otimes} I})$ and then since the previous bound holds for all $y \in\mathcal{L}^{\infty}_{F^{c}}$ and therefore for all $y \in \mathcal{L}^{\infty}_{J-F}$ we may apply the main result of \cite{NachtergaeleScholzWerner2013Local} to conclude that there exists an element $a' \in 1 \otimes \mathcal{B}((\ell^{2}(\mathbbm{N}))^{\otimes I})$ such that $\lVert b-a' \rVert \leq 2\epsilon$. We then have $\lVert a - a' \rVert \leq \lVert a - b \rVert + \lVert b - a' \rVert \leq 3\epsilon$. 
\end{proof}

While it seems natural that algebras localized in disjoint regions should commute in great generality, this very strong form of Haag duality has important consequences for the structure of a quasi-local algebra as shown in the following theorem.

\begin{thm}\label{trivialdhr} There is an equivalence of braided W*-tensor categories $\text{DHR}(\mathcal{L}^{\infty})\cong \text{Hilb}$.
\end{thm}

\begin{proof}
Let $X\in \text{DHR}_{0}(\mathcal{L}^{\infty})$, and choose $R$ to be the maximum of the localization lengths of $X$. Let $I$ be an interval with $|I|\ge R$, and let $\{b_{i}\}^{n}_{i=1}$ be a localizable basis. Since $\mathcal{L}^{\infty}$ satisfies Haag duality, then $X_{I}:=\{x\in X\ :\ [x,\mathcal{L}^{\infty}_{I^{c}}]=0\}=\text{span}\{b_{i}\mathcal{L}^{\infty}_{I}\}$ is a normal $\mathcal{L}^{\infty}_{I}\coloneq  \mathcal{B}((\ell^{2}(\mathbbm{N}))^{\otimes I})$-correspondence. But the only normal correspondences of $\mathcal{B}((\ell^{2}(\mathbbm{N}))^{\otimes I})$ are of the form $\mathcal{B}((\ell^{2}(\mathbbm{N}))^{\otimes I})\otimes H$, where $H$ is a separable multiplicity space. Then the tensor product decomposition gives $X\cong \mathcal{L}^{\infty}\otimes H$, where $H$ is the separable multiplicity.

Now suppose that $H$ is a separable Hilbert space, and consider $\mathcal{L}^{\infty}\otimes H$. To see this is realized as a DHR-bimodule, pick a unitary $U:\ell^{2}(\mathbbm{N})\cong \ell^{2}(\mathbbm{N})\otimes H$. Then consider the endomorphism of $\mathcal{B}(\ell^{2}(\mathbbm{N}))$ defined by $\rho(a):=U^{*}(a\otimes 1_{H})U$. Then the $\rho$-twisted correspondence $_{\rho}\mathcal{B}(K)$ is isomorphic to $\mathcal{B}(\ell^{2}(\mathbbm{N}))\otimes H$, but has a single basis element $1_{\mathcal{B}(\ell^{2}(\mathbbm{N}))}$. Extending to all of $\mathcal{L}^{\infty}$ gives the desired result.


\end{proof}

\subsection{Stabilization of quasi-local algebras}

Now let $A$ be a quasi-local algebra over $\mathbbm{Z}$.  We define the stabilization of $A$ (denoted by $A^{\infty}$) as follows:

$$A^{\infty}_{I}:= A_{I}\overline{\otimes} \mathcal{B}^{\infty}_{I}=A_{I}\overline{\otimes} \mathcal{L}(\ell^{2}(\mathbb{N}))^{\overline{\otimes} I}.$$

\noindent Taking the C*-algebraic colimit of the $A^{\infty}_{I}$ yields a quasi-local C*-algebra $A^{\infty}$, which we call the stabilization. A quasi-local algebra $A$ is called \textit{stable} if $A^{\infty}$ is bounded spread isomorphic to $A$.

Clearly the stabilization of any quasi-local algebra is stable, hence the terminology. From a condensed matter perspective, we are adding an infinite-dimensional space of ancilla at every site. At this level of generality, it is not yet clear that the stabilization of a quasi-local algebra satisfies Haag duality in general. This is an interesting question worth pursuing in its own right, though there is a simple way to address this for a class of quasi-local algebras that includes fusion spin chains, which are the main objects of study in this paper.

If $\phi$ is a faithful state on a C*-algebra, then we can define an inner product on $A$ via $\langle a, b\rangle_{\phi} \coloneqq \phi(b*a)$. Completing $A$ with respect to this inner product yields the space $L^{2}(A,\phi)$. 

\begin{defn}
    A satisfies Haag duality with respect to a faithful state $\phi$ if for $\widetilde{A} = A''$ in $L^{2}(A,\phi)$ we have $A_{I^{c}}' \cap \widetilde{A} = A_{I}$. 
\end{defn}
\begin{thm}
Suppose $A$ is a quasi-local algebra which satisfies Haag duality with respect to a faithful state $\phi$. Then $A^{\infty}$ satisfies Haag duality. 
\end{thm}

\begin{proof}
Take $x \in Z_{A^{\infty}}(A^{\infty}_{I^{c}})$ and let $R_{\phi}$ be the right slice map associated with $\phi$ (see \cite{SliceMapProblem} for a good introduction). We have $R_{\phi}(x) \in (\mathcal{L}^{\infty}_{I^{c}})^{\prime} \cap \mathcal{L}^{\infty} = Z_{\mathcal{L}^{\infty}}(\mathcal{L}^{\infty}_{I^{c}})$. Since $\mathcal{L}^{\infty}$ satisfies Haag Duality this implies $R_{\phi}(x) \in \mathcal{L}^{\infty}_{I}$. Since $A$ is nuclear it has Wasserman's property S, which means $(A,\mathcal{L}^{\infty},\mathcal{L}^{\infty}_{I})$ has the slice property \cite{PathologyInIdealSpace}. Thus $x \in A \otimes \mathcal{L}^{\infty}_{I}$. For any state $\psi \in (\mathcal{L}^{\infty})_{*}$ we have that $L_{\psi}(x) \in A_{I^{c}}'$. Consider the image of $x$ in $L^{2}(A,\phi)$, which lives in $\widetilde{A}\overline{\otimes} \mathcal{L}^{\infty}_{I}$. By assumption then $x \in A_{I} \overline{\otimes} \mathcal{L}^{\infty}_{I}$. Using uniqueness of tensor products in the finite case we conclude then that $x \in A_{I}^{\infty}$.

\end{proof}

\begin{cor}
For any fusion spin chain $A(\mathcal{C},X)$, $A(\mathcal{C},X)^{\infty}$ satisfies Haag duality.
\end{cor}

\begin{proof}
Standard subfactor-theoretic results show that for any fusion spin chain, there is a canonical faithful tracial state on $A(\mathcal{C},X)$ given on the local algebras by the normalized categorical trace which satisfies Haag duality (e.g. Ocneanu compactness) \cite{Bischoff2025Distortion, jones2025holographybulkboundarylocaltopological}, and for the most general version \cite{DasGhoshGhoshJones2025UnitaryConnections}.
\end{proof}

\section{Stable equivalence for fusion spin chains}

This section contains the main results of the paper. Let $\mathcal{C}$ be an indecomposable unitary multi-fusion category and let $\text{Hilb}(\mathcal{C})$ be the W*-tensor category introduced earlier. There is a distinguished object in $\text{Hilb}(\mathcal{C})$

$$W:=\bigoplus_{b\in \text{Irr}(\mathcal{C})} \ell^{2}(\mathbbm{N})\boxtimes b\ ,$$

\noindent which has the property that $W\otimes W\cong W$. It is also characterized by the property that $W\oplus U\cong W$ for any $U\in \text{Hilb}(\mathcal{C})$. If $\mathcal{C}$ is fusion, this object is also absorbing under $\otimes$.

Now, suppose $\mathcal{C}$ and $\mathcal{D}$ are indecomposable unitary multifusion categories, and that there is an invertible $\mathcal{D}$-$\mathcal{C}$ bimodule category $\mathcal{M}$. We construct the \textit{Morita context} for $\mathcal{D}$ and $\mathcal{C}$ by taking the right $\mathcal{C}$-module category  $\widetilde{\mathcal{M}}:=\mathcal{C}\oplus \mathcal{M}$, and building the multifusion category

$$\mathcal{E}:=\text{End}_{\mathcal{C}}(\widetilde{M}).$$

From the direct sum decomposition of module categories $\widetilde{\mathcal{M}}:=\mathcal{C}\oplus \mathcal{M}$, we see that $\mathcal{E}$ has the matrix decomposition

$$\mathcal{E}\cong \begin{bmatrix}
\text{End}_{\mathcal{C}}(\mathcal{C}) & \text{Fun}_{ \mathcal{C}}(\mathcal{M}, \mathcal{C}) \\
\text{Fun}_{ \mathcal{C}}(\mathcal{C}, \mathcal{M}) & \text{End}_{\mathcal{C}}(\mathcal{M}) 
\end{bmatrix}\cong \begin{bmatrix}
\mathcal{C} & \overline{\mathcal{M}} \\
\mathcal{M} & \mathcal{D} 
\end{bmatrix}$$

\bigskip

\noindent In particular, we see that $\mathcal{E}$ is an indecomposable multifusion category with diagonal components $\mathcal{E}_{11}\cong \mathcal{C}$ and $\mathcal{E}_{22}\cong \mathcal{D}$, with $\mathcal{E}_{21}\cong \mathcal{M}$ as a C*-category, and $\mathcal{E}_{12}$ equivalent to the inverse invertible bimodule category $\overline{\mathcal{M}}$. We note that \textit{every} indecomposable $2\times 2$ multifusion category has this form.

Now, for $1\le i,j\le 2$, set

$$W_{ij}:=\bigoplus_{b\in \text{Irr}(\mathcal{E}_{ij})} \ell^{2}(\mathbbm{N})\boxtimes b$$

\noindent Then since any simple object $a\in \mathcal{E}_{11}$ appears as a subalgebra of $b\otimes c$ for some $b\in \mathcal{E}_{12}$ and $c\in \mathcal{E}_{21}$ (for example, $a$ is a summand of $(a\otimes m)\otimes m^{*}$ for any $m\in \mathcal{E}_{12}$), we can choose isomorphisms 

$$\Phi: W_{11}\cong W_{12}\otimes W_{21}$$

and similarly

$$\Psi:W_{21}\otimes W_{12}\cong W_{22}.$$

\noindent We are now ready to prove one of our main technical results.

\begin{thm}\label{thm:MoritaEquivimpliesStableEquiv} Suppose $\mathcal{C}$ and $\mathcal{D}$ are Morita equivalent unitary indecomposable multifusion categories (or equivalently, $\mathcal{Z}(\mathcal{C})\cong \mathcal{Z}(\mathcal{D})$ as unitary braided tensor categories). Let $X\in \mathcal{C}$ and $Y\in \mathcal{D}$ be strong tensor generators. Then the stabilized fusion spin chains $A(\mathcal{C},X)^{\infty}$ and $A(\mathcal{D},Y)^{\infty}$ are bounded spread isomorphic. In particular, any anyon chain is stably equivalent to $\mathcal{L}^{\infty}$.
\end{thm}

\begin{proof}

Let $n$ be such that $X^{\otimes n}$ and $Y^{\otimes n}$ both contain all simple objects in $\mathcal{C}$ and $\mathcal{D}$ respectively.  

Pick a partition of $\mathbbm{Z}$ into intervals $\{I_{i}\}_{i\in \mathbbm{Z}}$ of length $n$, ordered left to right by index $i$. Then for any local algebras we have

\begin{align*}
A(\mathcal{C},X)^{\infty}_{I}&=A(\mathcal{C},X)_{I}\otimes B(K)^{\overline{\otimes} I}\\
&\cong \text{End}_{\mathcal{C}}(X^{\otimes I})\otimes B(K^{\otimes I})\\
&\cong \text{End}_{\text{Hilb}(\mathcal{C})}(K^{\otimes I}\boxtimes X^{\otimes I})\\
\end{align*}

In particular $$A(\mathcal{C},X)^{\infty}_{I_{i}}\cong \text{End}_{\text{Hilb}(\mathcal{C})}(W_{11}).$$

\medskip

\noindent More generally, we see that if $J=[a,b]\subseteq \mathbbm{Z}$ is an interval, we can define the associated physical interval $I_{J}:=I_{a}\cup I_{a+1}\cup \dots I_{b}$ the interval in $\mathbbm{Z}$ of length $n(b-a+1)$. The algebra

$$A_{I_{J}}\cong \text{End}_{\text{Hilb}(\mathcal{C})} (W^{\otimes J}_{C})=A(\text{Hilb}(\mathcal{E}), W_{11})_{J}$$

\noindent and for $J\subseteq J^{\prime}$, we have that the inclusion $A_{I_{J}}\hookrightarrow A_{I_{J^{\prime}}}$ is given by 

$$f\mapsto 1_{W_{\mathcal{D}}}\otimes \dots \otimes f \otimes \dots \otimes 1_{W_{\mathcal{D}}}$$

\noindent Thus after coarse-graining with this interval partition, we see that at the level of C*-algebras $$A^{\infty}(\mathcal{C},X)\cong A(\text{Hilb}(\mathcal{E}),W_{11})$$

\noindent Similarly, if we use the same coarse-graining by interval partitions, we have  

\medskip

$$A(\mathcal{D},Y)^{\infty}_{I_{J}}\cong \text{End}_{\text{Hilb}(\mathcal{D})} (W^{\otimes J}_{D})\cong A(\text{Hilb}(\mathcal{E}),W_{22})_{J} $$

\medskip

\noindent and an isomorphism of quasi-local algebras

\medskip

$$A^{\infty}(\mathcal{D},Y)\cong A(\text{Hilb}(\mathcal{E}),W_{11})$$

\noindent 

We now build a bounded spread isomorphism $\alpha: A^{\infty}(\mathcal{C}, X)\cong A^{\infty}(\mathcal{D}, Y)$, which is simplest to define graphically for $a\in A^{\infty}(\mathcal{C}, X)_{I_{J}}\cong A(\text{Hilb}(\mathcal{E}), W_{11})_{J}$, by

$$
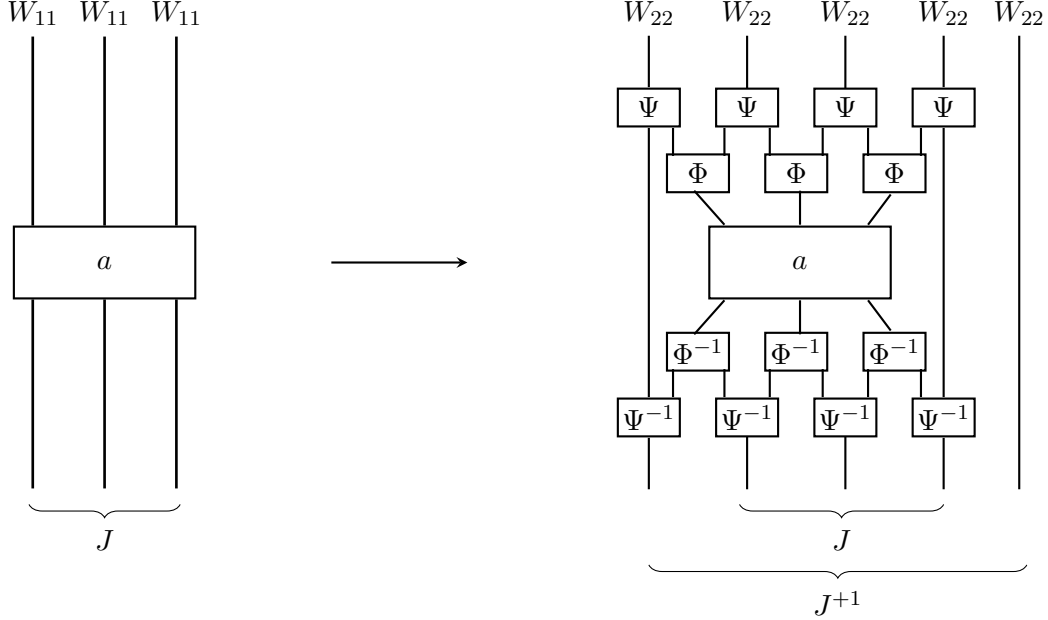
\begin{tikzpicture}[x=1cm,y=1cm,>=stealth]
\tikzset{
  wire/.style={line width=1pt},
  bigbox/.style={draw, thick, minimum width=2.4cm, minimum height=0.95cm},
  smallbox/.style={draw, thick, minimum width=0.82cm, minimum height=0.50cm, inner sep=1pt},
  lab/.style={font=\small}
}

\node[bigbox] (fL) at (0,0) {$a$};

\foreach \x in {-0.95,0,0.95}{
  \draw[wire] ($(fL.north)+(\x,0)$) -- ++(0,2.5);
  \draw[wire] ($(fL.south)+(\x,0)$) -- ++(0,-2.5);
}

    \node at (-0.95,3.3) {$W_{11}$};
    \node at (0,3.3) {$W_{11}$};
    \node at (0.95,3.3) {$W_{11}$};
      \draw[decorate,decoration={brace,mirror,amplitude=5pt}]
    (-1,-3.2) -- (1,-3.2)
    node[midway,below=6pt] {$J$};

\draw[->, thick] (3.0,0) -- (4.8,0);

\node[bigbox] (fR) at (9.2,0) {$a$};

\def\xA{7.2}
\def\xB{8.5}
\def\xC{9.8}
\def\xD{11.1}

\def\xAB{7.85}
\def\xBC{9.15}
\def\xCD{10.45}

\def\yPsi{2.05}
\def\yPhi{1.18}
\def\yIphi{-1.18}
\def\yIpsi{-2.05}

\draw[thick] (7.2,2.32) -- (7.2,3);
\draw[thick] (7.2,1.78) -- (7.2,-1.78);
\draw[thick] (7.2,-2.32) -- (7.2,-3);

\draw[thick] (11.1,2.32) -- (11.1,3);
\draw[thick] (11.1,1.78) -- (11.1,-1.78);
\draw[thick] (11.1,-2.32) -- (11.1,-3);

\draw[thick] (7.52,1.78) -- (7.52,1.41);

\draw[thick] (7.52,-1.78) -- (7.52,-1.41);

\draw[thick] (8.8,1.78) -- (8.8,1.41);
\draw[thick] (9.2,0.5) -- (9.2,0.95);
\draw[thick] (9.2,-0.5) -- (9.2,-0.95);
\draw[thick] (8.8,-1.78) -- (8.8,-1.41);

\draw[thick] (8.2,1.78) -- (8.2,1.41);
\draw[thick] (7.8,0.95) -- (8.2,0.5);
\draw[thick] (8.2,-0.5) -- (7.8,-0.95);
\draw[thick] (8.2,-1.78) -- (8.2,-1.41);

\draw[thick] (8.5,2.3) -- (8.5,3);
\draw[thick] (8.5,-2.3) -- (8.5,-3);

\draw[thick] (9.8,2.3) -- (9.8,3);
\draw[thick] (9.5,1.78) -- (9.5,1.41);
\draw[thick] (9.5,-1.78) -- (9.5,-1.41);
\draw[thick] (9.8,-2.3) -- (9.8,-3);

\draw[thick] (10.1,1.78) -- (10.1,1.41);

\draw[thick] (10.1,-1.78) -- (10.1,-1.41);

\draw[thick] (10.8,1.78) -- (10.8,1.41);

\draw[thick] (10.8,-1.78) -- (10.8,-1.41);

\draw[thick] (10.1,0.5) -- (10.4,0.9);

\draw[thick] (10.1,-0.5) -- (10.4,-0.9);

\draw[thick] (12.1,-3) -- (12.1,3);

\node[smallbox] (psi1) at (\xA,\yPsi) {$\Psi$};
\node[smallbox] (psi2) at (\xB,\yPsi) {$\Psi$};
\node[smallbox] (psi3) at (\xC,\yPsi) {$\Psi$};
\node[smallbox] (psi4) at (\xD,\yPsi) {$\Psi$};

\node[smallbox] (phi1) at (\xAB,\yPhi) {$\Phi$};
\node[smallbox] (phi2) at (\xBC,\yPhi) {$\Phi$};
\node[smallbox] (phi3) at (\xCD,\yPhi) {$\Phi$};

\node[smallbox] (iphi1) at (\xAB,\yIphi) {$\Phi^{-1}$};
\node[smallbox] (iphi2) at (\xBC,\yIphi) {$\Phi^{-1}$};
\node[smallbox] (iphi3) at (\xCD,\yIphi) {$\Phi^{-1}$};

\node[smallbox] (ipsi1) at (\xA,\yIpsi) {$\Psi^{-1}$};
\node[smallbox] (ipsi2) at (\xB,\yIpsi) {$\Psi^{-1}$};
\node[smallbox] (ipsi3) at (\xC,\yIpsi) {$\Psi^{-1}$};
\node[smallbox] (ipsi4) at (\xD,\yIpsi) {$\Psi^{-1}$};

 \node at (7.2,3.3) {$W_{22}$};
    \node at (8.5,3.3) {$W_{22}$};
    \node at (9.8,3.3) {$W_{22}$};
    \node at (11.1,3.3) {$W_{22}$};
    \node at (12.1,3.3) {$W_{22}$};
    \draw[decorate,decoration={brace,mirror,amplitude=5pt}]
    (8.4,-3.2) -- (11.1,-3.2)
    node[midway,below=6pt] {$J$};
      \draw[decorate,decoration={brace,mirror,amplitude=5pt}]
    (7.2,-4) -- (12.2,-4)
    node[midway,below=6pt] {$J^{+1}$};

\end{tikzpicture}
$$

\noindent with $\alpha(a)\in A(\text{Hilb}(\mathcal{E}), W_{22})_{J^{+1}}\cong A(\text{Hilb}(\mathcal{D}), W_{\mathcal{D}})_{J^{+1}}\cong A^{\infty}(\mathcal{D},Y)_{I_{J^{+1}}}.$

\medskip

This extends to a $*$-homomorphism with spread at most $|I_i|=n$. But in fact, this has an inverse, defined similarly but exchanging domain and range, and interchanging the roles $\Phi\leftrightarrow \Psi^{-1}$.

\end{proof}

\subsection{Factorization of stabilized anyon chains: Hilbert space level}\label{hilbspacelevel}

Anyon chains, in particular, have a spatial realization. We can phrase this quite simply in terms of a graph. If $\mathcal{C}\rightarrow \text{End}(\mathcal{M})$ and $X$ is in the centralizer, then the graph in question has vertices indexed by simple objects, i.e. $\text{Irr}(\mathcal{M})$, and the number of edges $a\rightarrow b$ is given by the dimension of the space $\mathcal{M}(a,X(b))$.

Now, if $G$ is an arbitrary oriented, directed graph with finitely many vertices labeled $0,1,\dots, t-1$, for each pair of vertices, we have a finite edge set $E_{i\rightarrow j}$. Then consider the finite-dimensional Hilbert space

$$H_{i,j}:=\mathbbm{C}[E_{i\rightarrow j}],$$

\noindent where the edges serve as an orthonormal basis. Then the local Hilbert space associated to the interval $I$ of length $m$ with boundary conditions $i,j$ is

$$H^{I}_{i,j}:=\bigoplus_{1\le i_{1}, \dots i_{m-1}\le m}H_{i,i_{1}}\otimes H_{i_{1},i_{2}}\otimes \dots \otimes H_{i_{m-2},i_{m-1}}\otimes H_{i_{m-1},j}$$

\noindent with local operator algebra the operators on $H^{I}:=\bigoplus_{i,j} H^{I}_{i,j}$ which preserve boundary conditions, i.e.

$$A_{I}:=\bigoplus_{i,j} \mathcal{B}(H_{i,j}).$$

We note that since each $H_{i,j}$ has an orthonormal basis of edges from $i$ to $j$, then $H^{I}_{i,j}$ has an orthonormal basis consisting of paths of length $m$ from $i$ to $j$, with the path $p=(e_1\dots e_{m})$ associated with the vector $e_{1}\otimes e_{2}\otimes \dots \otimes e_{m}\in H^{I}_{i,j}$. We denote the set of paths of length $m$ from $i\rightarrow j$ in $G$ by $G^{m}(i\rightarrow j)$. From here on out, we will assume that $G$ is strongly connected in the sense that there exists some $m$ such that for all $l\ge m$, $G^{l}(i\rightarrow j)\ne \varnothing$, or in other words, $H^{I}_{i,j}\ne 0$ for any interval $|I|\ge m$.

In this case, we have that for any $i,j$ and interval $I$ of length $2k\ge m$, 

$$W_{i,j}:=H^{I}_{i,j}\otimes K^{\otimes 2k}\bigoplus_{1\le i_{1}, \dots i_{k-1}\le n }H_{i,i_{1}} \otimes H_{i_{1},i_{2}}\otimes \dots \otimes H_{i_{m-2},i_{m-1}}\otimes H_{i_{2k-1},j}\otimes (\ell^{2}(\mathbb{N}))^{\otimes 2k}\ne 0.$$

\noindent  Now, let $W_{i,*}:= (\ell^{2}(\mathbb{N}))^{\otimes k}, W_{*,j}\cong (\ell^{2}(\mathbb{N}))^{\otimes k}$.

Then 
$$W_{i,j}:=\mathbbm{C}[G^{2k}(i\rightarrow j)]\otimes (\ell^{2}(\mathbb{N}))^{\otimes 2k}\cong \bigoplus_{p\in G^{2k}(i\rightarrow j)} (\ell^{2}(\mathbb{N}))^{\otimes 2k}$$

Now, since $W_{i,j}$ and $W_{i,*}\otimes W_{j}\cong (\ell^{2}(\mathbb{N}))^{\otimes 2k}$ are separable infinite-dimensional Hilbert spaces, we can choose an isomorphism 

$$\Phi_{i,j}:W_{i,j}\cong W_{i,*}\otimes W_{*,j}.$$

\noindent We can make this more explicit. Suppose $K:=\ell^{2}(\mathbbm{N})$ is the Hilbert space for a harmonic oscillator, with basis $\{|n\rangle\}_{n\in \mathbbm{N}}$. Then we have an orthonormal basis

$$\{|p,n_{1}, \dots n_{2k}\rangle\ :\ p\in G^{2k}(i\rightarrow j), n_{r}\in \mathbbm{N} \}.$$

\noindent On the other hand $W_{i,*}\otimes W_{*,j}\cong K^{\otimes 2k}$ has a basis given by $$\{ |n_{1}, \dots, n_{2k}\rangle\ :\ n_{i}\in \mathbbm{N}\}$$

\noindent If we pick an identification $G^{2k}(i \rightarrow j)\cong \{0,1,\dots, l-1\}$ then we can define the following unitary on basis vectors

$$\Phi_{i,j}|p,n_{1},\dots, n_{2k}\rangle:=|ln_{1}+p,\ n_{2},\ \dots,\  n_{2k}\rangle$$

\noindent which has inverse

$$\widetilde{\Phi}^{\dagger}_{i,j}|n_{1},\ \dots ,\ n_{2k}\rangle:=|n_{1}\ \text{mod}\ l\ , \frac{n_{1}-(n_{1} \text{mod}\ l)}{l}, n_{2},\ \dots,\ n_{2k}\rangle.$$

The final ingredient we need is an isomorphism $$\Psi:\bigoplus_{j} W_{*,j}\otimes W_{j,*}\cong K^{\otimes 2k} $$

\noindent But $\bigoplus_{j} W_{*,j}\otimes W_{j,*}\cong \bigoplus_{j} K^{\otimes 2k}\cong \mathbbm{C}^{t}\otimes K^{\otimes 2k}$, and thus has a natural basis

$$\{|j,n_{1},\dots n_{2k}\rangle\ :\ 0\le j < t,\ n_{i}\in \mathbbm{N} \},$$

\noindent so that $\Psi$ can be defined by a formula similar to $\Phi_{i,j}$,

$$\Psi|j,n_{1},\dots, n_{2k}\rangle:=|tn_{1}+j,\ n_{2},\ \dots,\  n_{2k}\rangle.$$

Now, putting this all together, suppose we have an interval $I$ in $\mathbbm{Z}$ of length $n\cdot 2k$, which we think of as $n$ consecutive intervals each of length 2k. Then we have a composition of unitary isomorphisms

\begin{tikzcd}
\displaystyle H^{I}_{i,j}=\bigoplus_{0\le i_{1},\dots i_{n-1}\le t-1} W_{i,i_{1}}\otimes W_{i_{1},i_{2}}\otimes \dots \otimes W_{i_{n-1},j}\arrow[swap, ddd,"\ \ \ \displaystyle\bigoplus_{0\le i_{1},\dots i_{n-1}\le t-1}\Phi_{i,i_{1}}\otimes \Phi_{i_{1},i_{2}}\otimes \dots\ \otimes \Phi_{i_{n-1},j}\ \ \ "] \\
\\
\\
\displaystyle \bigoplus_{0\le i_{1},\dots i_{n-1}\le t-1}W_{i,*}\otimes W_{*,i_1}\otimes W_{i_1,*}\otimes W_{*,i_{2}}\otimes \dots \otimes W_{i_{n-1},*}\otimes W_{*, j}\arrow[swap, d,"\cong\ \ \ \ "]\\
 W_{i,*}\otimes \left(\bigoplus_{i_1} W_{*,i_1}\otimes W_{i_1,*}\right)\otimes \dots \otimes \left(\bigoplus_{0\le i_{n-1}\le t-1}W_{*,i_{n-1}}\right)\otimes W_{*,j} \arrow[swap, ddd, "\displaystyle 1_{W_{i,*}} \otimes \Psi \otimes\ \dots\ \otimes \Psi\otimes 1_{W_{*,j}}\ \ \ "]\\
 \\
 \\
 W_{i,*}\otimes K^{\otimes 2k}\otimes\ \dots\ \otimes K^{\otimes 2k}\otimes W_{*,j} 
\end{tikzcd}

\bigskip

\noindent We call this unitary composite $$\Lambda^{I}_{i,j}: H^{I}_{i,j}\cong W_{i,*}\otimes K^{\otimes (n-1)2k} \otimes W_{*,j}.$$

If we set $J$ to be the interval consisting of the middle $(n-1)2k$ sites of $I$, then the middle Hilbert space $K^{\otimes (n-1)2k}= K^{\otimes J}$ carries the standard, faithful representation of $\mathcal{L}^{\infty}_{J}$. In particular, if we consider the representation 

$$\pi_{i,j}:\mathcal{L}^{\infty}_{J}=\mathcal{B}(K^{\otimes J})\rightarrow \mathcal{B}(W_{i,*}\otimes K^{\otimes (n-1)2k} \otimes W_{*,j})$$

$$\alpha_{J}(a):=\bigoplus_{i,j}(\Lambda^{I}_{i,j})^{\dagger}\circ\left(1_{W_{i,*}}\otimes a \otimes 1_{W_{*,j}}\right) \circ\Lambda^{I}_{i,j}\in A_{I}=A_{J^{+k}}.$$

\noindent It is straightforward to check by hand that the assignment on local algebras is consistent (indeed, one can explicitly construct the inverse) and defines a bounded spread isomorphism $\alpha:\mathcal{L}^{\infty}\cong A^{\infty}$, as desired.

\section{Classification results}

\subsection{Classification of fusion spin chains up to stable equivalence}

Let $\mathcal{C}$ be a braided W*-tensor category and $\mathcal{D}\subseteq \mathcal{C}$ a full subcategory. We define the M\"{u}ger centralizer $\mathcal{D}^{\prime}\cap \mathcal{C}$ as the full tensor subcategory consisting of objects $X$ such that $\sigma_{Y,X}\circ \sigma_{X,Y}=1_{X\otimes Y}$ for all $Y\in \mathcal{D}$. $\mathcal{C}$ is called \textit{non-degenerately braided} if $\mathcal{C}'\cap \mathcal{C}$ consists of direct sums of the unit object $\mathbbm{1}$. An object $Y\in \mathcal{C}$ is disjoint from the subcategory $\mathcal{D}$ if $\text{Hom}_{\mathcal{C}}(X,Y)=0$ for all $X\in \mathcal{D}$.

We also recall that a non-degenerately braided unitary fusion category has a unique unitary spherical structure, hence can be viewed as a unitary modular tensor category \cite{Rowell2006UMTC}. We will make use of the graphical calculus for modular tensor categories \cite{BakalovKirillov2001}.


\begin{lem}\label{centralizerlemma}
Let $\mathcal{C}$ be a braided W*-tensor category closed under countable direct sums, and $\mathcal{D}\subseteq \mathcal{C}$ a full, non-degenerately braided fusion subcategory. If there is an object $Y\in \mathcal{C}$ disjoint from $\mathcal{D}$, then there is a non-zero object in $\mathcal{D}'\cap\mathcal{C}$ disjoint from $\mathcal{D}$.
\end{lem}

\begin{proof}
Let $Y\in \mathcal{C}$. Then for each $a\in \text{Irr}(\mathcal{D})$, we have the loop operator

$$\rho^{Y}_{a}:=\begin{tikzpicture}[
    baseline=-0.5ex,
    strand/.style={
        draw=black,
        line width=0.8pt,
        line cap=round,
        line join=round
    },
    over/.style={
        preaction={
            draw=white,
            line width=3.2pt,
            line cap=butt
        }
    }
]
    \def\rx{1}
    \def\ry{0.38}
    \def\gap{0.18}

    \draw[strand]
        (\rx,0)
        arc[start angle=0,end angle=360,
            x radius=\rx,y radius=\ry];

    \draw[strand]
        (0,-1.25) -- (0,1.25);

    \draw[strand,over]
        (0,{\ry-\gap}) -- (0,{\ry+\gap});

    \draw[strand,over]
        ({\rx*cos(245)},{\ry*sin(245)})
        arc[start angle=245,end angle=295,
            x radius=\rx,y radius=\ry];

    \node[right=3pt] at (0,1.05) {$Y$};
    \node[right=3pt] at (\rx-2.2,0.5) {$a$};
\end{tikzpicture}\in \text{End}_{\mathcal{C}}(Y)$$

The assignment $a\mapsto \rho^{Y}_{a}$ furnishes a homomorphism of the fusion ring of $\mathcal{D}$ into $\text{End}(Y)$ for any object $Y\in \mathcal{C}$. We consider the virtual regular element

$$\rho^{Y}:=\frac{1}{\sum_{a\in \text{Irr}(\mathcal{D})d^{2}_{a}}}\sum d_{a} \rho^{Y}_{a} $$

Then $\rho^{Y}$ is an orthogonal projection in $\text{End}_\mathcal{C}(Y)$. We denote the image of this idempotent $P_{Y}\le Y$. The standard graphical calculus arguments (for example, the proof of \cite[Lemma 3.3]{BaerenzBarrett2018}) show that $P_{Y}\in \mathcal{D}'\cap \mathcal{C}$. 

Now, if $Y$ is nonzero and disjoint from $\mathcal{D}$, then so is $P_{Y}$. However, $P_{Y}$ may be zero. We claim that at least one of the objects $P_{a\otimes Y}$ as $a$ ranges over the simple objects of $\mathcal{D}$ is nonzero. Indeed, if $P_{a\otimes Y}=0$, then taking the partial trace (and using  \cite[Lemma 3.1.4]{BakalovKirillov2001}) shows that $\sum_{b\in \text{Irr}(\mathcal{D})} \widetilde{s}_{a,b}=0$, where the $\text{rank}(\mathcal{D})\times \text{rank}(\mathcal{D})$ matrix $\widetilde{s}$ defined by $\widetilde{s}_{a,b}=\text{Tr}(1_{a\otimes \overline{b}})$ is the unnormalized S-matrix (following \cite[Chapter 3]{BakalovKirillov2001}). But since $\mathcal{D}$ is modular, the matrix $\widetilde{s}$ is invertible, so we cannot have $\sum_{b\in \text{Irr}(\mathcal{D})} \widetilde{s}_{a,b}=0$ for all $a$, as the vector $\begin{bmatrix}
1\\
1\\
\dots\\
1   
\end{bmatrix}$ would be in the kernel.
\end{proof}

\begin{thm}\label{StabilizedDHR}
Suppose $A:=A(\mathcal{C},X)$, where $\mathcal{C}$ is an indecomposable unitary multi-fusion category and $X$ is a strong tensor generating object. Then $\text{DHR}(A^{\infty})\cong \text{Hilb}(\text{DHR}(A))\cong \text{Hilb}(\mathcal{Z}(\mathcal{C}))$.
\end{thm}

\begin{proof}

In this case, since each $A_{I}$ is a finite-dimensional algebra and colimits over directed sets commute with the maximal tensor product of C*-algebras \cite{MR2391387}, we have

\begin{align*} A^{\infty}=\text{colim}_{I}\ \left(A_{I}\overline{\otimes} \mathcal{L}^{\infty}_{I}\right)&=\text{colim}_{I}\ \left(A_{I}\otimes_{\text{max}} \mathcal{L}^{\infty}_{I}\right)\\
&\cong (\text{colim}_{I} \ A_{I})\otimes_{\text{max}} (\text{colim}_{I} \ \mathcal{L}^{\infty}_{I})\\
&\cong A\otimes_{\text{max}} \mathcal{L}^{\infty}\\
&\cong A\otimes_{\text{min}} \mathcal{L}^{\infty}
\end{align*}

\medskip

\noindent where, in the last line, we are using that $A$ is AF, hence a nuclear C*-algebra. In particular, for every $X\in \text{Corr}(A)$, there is a well-defined correspondence $$X^{\infty}:=X\otimes \mathcal{L}^{\infty}\in \text{Corr}(A\otimes \mathcal{L}^{\infty})=\text{Corr}(A^{\infty}),$$

\noindent which is the completion of the linear tensor product $X\otimes_{\mathbbm{C}} \mathcal{L}^{\infty}$ with respect to the $A^{\infty}$-valued inner product on simple tensors given by

$$\langle a_{1}\otimes b_{1}\ |\ a_{2}\otimes b_{2}\rangle:=\langle a_{1}\ |\ a_{2}\rangle \otimes b^{\dagger}_{1}b_{2}.$$

Now, for any $X\in \text{DHR}(A)=\text{DHR}_{0}(A)$, we see that $X^{\infty}\in \text{DHR}_{0}(A^{\infty})$, since the projective localized bases for $X$ are clearly locally normal (since they only operate on the locally finite-dimensional tensor factors).

By using projective bases, we clearly have natural isomorphisms $(X\boxtimes_{A}Y)^{\infty}\cong X^{\infty}\boxtimes_{A^{\infty}} Y^{\infty}$, and thus we have a natural, braided tensor inclusion $\text{DHR}(A)\hookrightarrow \text{DHR}_{0}(A^{\infty})$. But since $\text{DHR}(A)$ is a non-degenerately braided fusion category, this is necessarily fully faithful by \cite{MR1990929}.

Now, pick any Lagrangian algebra object $L\in \mathcal{Z}(\mathcal{C})\cong \text{DHR}(A)\subseteq \text{DHR}_{0}(A^{\infty})$. Because realization commutes with colimits \cite[Theorem 4.6]{chen2022ktheoretic}, we have a strictly locality-preserving isomorphism $(A\rtimes L)^{\infty}\cong A^{\infty}\rtimes L$ where $A^{\infty}\rtimes L$ denotes the realization of the Q-system L \cite{MR4419534}. But $A\rtimes L$ is an anyon chain, and thus by the results of the previous section, we have a bounded spread isomorphism $\mathcal{L}^{\infty}\cong (A\rtimes L)^{\infty}\cong A^{\infty}\rtimes L$, hence $\text{DHR}(A^{\infty}\rtimes L)\cong \text{Hilb}$. Since $A^{\infty}$ is stable ($A^{\infty}\otimes\mathcal{L}^{\infty}\cong A^{\infty}$ by a spread-zero isomorphism), $\text{DHR}(A^{\infty})$ admits countable direct sums for objects with a common localization scale. In particular, if $\text{DHR}(A^{\infty})$ has no object disjoint from $Z(\mathcal{C})\cong \text{DHR}(A)\hookrightarrow \text{DHR}(A^{\infty})$, then since $\text{DHR}(A)$ is fusion, we can choose a uniformly bounded localization scale for all objects, $\text{DHR}(A^{\infty})\cong \text{Hilb}(\text{DHR}(A))$. 

If there were some object in $\text{DHR}(A^{\infty})$ disjoint from $\text{DHR}(A)$, then by Lemma \ref{centralizerlemma}, there is a non-zero object $Y$ disjoint from $\text{DHR}(A)$ which also centralizes it. Then the free $L$-module $Y\boxtimes_{A^{\infty}}L$ is local, hence defines a DHR bimodule of $\mathcal{L}^{\infty}$. On the other hand, this module is nontrival as an $L$-module, since 
$$\text{Hom}_{\text{DHR}(\mathcal{L}^{\infty})}(\mathcal{L}^{\infty}, Y\boxtimes_{A^{\infty}} L)\cong \text{Hom}_{\text{DHR}(A^{\infty}_{L})}(L, Y\boxtimes_{A^{\infty}} L)$$
$$\cong \text{Hom}_{\text{DHR}(A^{\infty})}(L, Y\boxtimes_{A^{\infty}}L)\cong \text{Hom}_{\text{DHR}(A^{\infty})}(L\boxtimes_{A^{\infty}} L, Y) =0.$$

The first isomorphism is a standard result about free $L$ modules (for example, see \cite{MR3687214}), the second isomorphism uses self-duality of the Lagrangian algebra $L$, and the final equality uses the disjointness of $Y$ from $\text{DHR}(A)$. In particular, the free module $Y\boxtimes_{A^{\infty}}L\in \text{DHR}(\mathcal{L}^{\infty})$ has no non-zero morphisms to the trivial bimodule, contradicting Theorem \ref{trivialdhr}.


\end{proof}

\begin{cor}\label{cor:realizesymmetriesonLinf}
If $B\subseteq A$ is an anyon model with symmetry category $\mathcal{C}$, then $B^{\infty}\subseteq A^{\infty}$ is a physical boundary subalgebra with symmetry category $\mathcal{C}$. Using the bounded spread isomorphism $A^{\infty}\cong \mathcal{L}^{\infty}$ from Theorem \ref{thm:MoritaEquivimpliesStableEquiv}, every fusion category can be realized by symmetries of  $\mathcal{L}^{\infty}$.  
\end{cor}

This naturally suggests the following problem:

\begin{quest} How many ways can a fusion category be realized in $\mathcal{L}^{\infty}$?
\end{quest}

While we do not provide a complete answer to this question, in the next section we show that all the symmetry realizations arising from stabilized anyon chains are equivalent for a fixed symmetry category. First we point out the following corollary.

\begin{cor}\label{cor:fusionspinchains} The fusion spin chains $A(\mathcal{C},X)$ and $A(\mathcal{D},Y)$ are stably equivalent if and only if $\mathcal{Z}(\mathcal{C})\cong \mathcal{Z}(\mathcal{D})$ as braided unitary fusion categories.
\end{cor}

\subsection{Stable equivalence of anyon models}

We now review some basic facts about Lagrangian algebras and indecomposable multi-fusion categories (see for example \cite{MR3242743, Bischoff2025Distortion, MR3406516}). Let $\mathcal{F}$ be an indecomposable unitary $n\times n$ multifusion category. Each of the fusion categories $\mathcal{F}_{ii}$ are pairwise Morita equivalent, with $\mathcal{F}_{ij}$ an invertible $\mathcal{F}_{ii}-\mathcal{F}_{jj}$ bimodule category (which in particular, are indecomposable as left and right $\mathcal{F}_{ii}$ and $\mathcal{F}_{jj}$ module categories, respectively).

In particular, if we let $\mathbbm{1}_{i}$ be the units of the diagonal multifusion categories, then  for each i, the functor 

$$F_{i}:\mathcal{Z}(\mathcal{F})\cong \mathcal{Z}(\mathcal{F}_{ii}),$$

$$F_{i}(Z,\sigma):=(Z_{ii},\sigma|_{\mathcal{F}_{ii}})$$

\noindent is a (unitary) braided equivalence.

For each $j$, there is a nice representation of the Lagrangian algebra $L_{ij}\in \mathcal{Z}(\mathcal{F}_{ii})$ associated to the indecomposable left $\mathcal{F}_{ii}$ module category $\mathcal{F}_{ij}$, which as an algebra object in $\mathcal{F}_{ii}$ is just

$$L_{ij}:=\bigoplus_{X\in \mathcal{E}_{ij}} X\otimes \overline{X} $$

\noindent with (dual) Q-system normalized multiplication structure

\[
m_{ij}
:=
\bigoplus_{X\in\mathrm{Irr}(\mathcal{F}_{ij})}
\frac{\sqrt{\text{dim}(\mathcal{E}_{ij})}}{\sqrt{d_X}}\ \ \ \ \,
\begin{tikzpicture}[baseline={(current bounding box.center)}, scale=0.7]
    \draw[thick]
        (0,.5) node[below] {$\scriptstyle \overline{X}$}
        -- (0,1) arc (180:0:1) -- (2,.5) node[below] {$\scriptstyle X$}

        (-1,.5) node[below] {$\scriptstyle X$}
        -- (-1,1) arc (180:135:2)
        to[in=-90,out=45] (.5,3.5) -- (.5,4) node[above] {$\scriptstyle X$}

        (3,.5) node[below] {$\scriptstyle \overline{X}$}
        -- (3,1) arc (0:45:2)
        to[in=-90,out=135] (1.5,3.5) -- (1.5,4) node[above] {$\scriptstyle \overline{X}$};
    \node at (1,2.7) {$\scriptstyle j$};
    \node at (3,3) {$\scriptstyle i$};
    \node at (1,1) {$\scriptstyle i$};
    \node at (-1,3) {$\scriptstyle i$};
\end{tikzpicture}
\]

\noindent and unit given by
$$\iota_{ij}=
\bigoplus_{X\in\text{Irr}(\mathcal{F}_{ij})}
\frac{\sqrt{d_X}}{\sqrt{\text{dim}(\mathcal{E}_{ij})}}\ \ \ \,
\begin{tikzpicture}[baseline={(current bounding box.center)}, scale=0.7]
    \draw[thick]
        (2,1) node[above] {$\scriptstyle \overline{X}$}
        -- (2,0) arc (360:180:1)
        (0,0) -- (0,1) node[above] {$\scriptstyle X$};
    \node at (1,.5) {$\scriptstyle j$};
    \node at (-0.8,0) {$\scriptstyle i$};
\end{tikzpicture}
$$

\noindent Here $i$ and $j$ label the regions, and are shorthand for including the projection $p_{i}\in \text{End}(\mathcal{F})_{i}$ onto the summand $\mathbbm{1}_{i}$ of the identity.
It is straightforward to check that this gives a dual Q-system structure on $L_{ij}$ (cf. \cite{MR4357481, 1004.1533}).

The half-braiding that lifts $L_{ij}$ to $\mathcal{Z}(\mathcal{F}_{ii})$ is the standard one

\begin{equation}\label{eqn:half_braiding}
\sigma_{L_{ij},W}
:=
\bigoplus_{Y,Z\in\mathrm{Irr}(\mathcal{F}_{ij})}
\sum_{\alpha\in \Lambda}
\frac{\sqrt{d_Z}}{\sqrt{d_Y}}\,
\begin{tikzpicture}[baseline={(current bounding box.center)},scale=0.7]
    \draw[thick]
        (0,0) node[below] {$\scriptstyle Y$}
            -- (0,4) node[above] {$\scriptstyle Z$}
        (2,0) node[below] {$\scriptstyle \overline{Y}$}
            -- (2,4) node[above] {$\scriptstyle \overline{Z}$}
        (0,2) to[out=110,in=270] (-1,4) node[above] {$\scriptstyle W$}
        (2,2) to[out=290,in=90] (3,0) node[below] {$\scriptstyle W$};

    \draw[thick, fill=white] (0,2) circle (0.07);
    \draw[thick, fill=white] (2,2) circle (0.07);

    \node[left]  at (0,2) {$\scriptstyle \alpha$};
    \node[right] at (2,2) {$\scriptstyle \alpha^{\bullet}$};
    \node at (1,2) {$\scriptstyle j$};
    \node at (-1,1) {$\scriptstyle i$};
    \node at (3,3) {$\scriptstyle i$};
    \node at (2.4,0.6) {$\scriptstyle i$};
    \node at (-0.4,3.4) {$\scriptstyle i$};
\end{tikzpicture}
\,.
\end{equation}

\noindent where $\Lambda$ is a basis for $\text{Hom}_{\mathcal{F}_{ij}}(Y, W\otimes Z)$ with $\beta^{\dagger}\circ \alpha=\delta_{\alpha,\beta} 1_{Y}$ for all $\alpha,\beta\in \Lambda$, and $\alpha^{\circ}$ is the partial rotation of $\alpha^{\dagger}$, $\alpha^{\bullet}:=(\text{ev}_{Y} \otimes 1_{\overline{Z}}) \circ  \alpha^{\dagger} \circ (1_{\overline{Y}}\otimes 1_{W}\otimes \text{coev}_{Z}).$

Under the equivalence $F^{-1}_{i}:\mathcal{Z}(\mathcal{F}_{ii})\cong \mathcal{Z}(\mathcal{F})$, 

$$L_{j}:=F^{-1}_{i}(L_{ij})\cong \bigoplus_{i} L_{ij}\in \bigoplus_{i} \mathcal{F}_{ii}\subseteq \mathcal{F},$$ with unitary half-braiding on all of $\mathcal{F}$ given by the same formula as above, but now $W$ can be an arbitrary object in $\mathcal{F}$, given on a homogeneous object $W\in \mathcal{F}_{ik}$ by 

$$
\sigma_{L_{j},W}
:=
\bigoplus_{Y\in\mathrm{Irr}(\mathcal{F}_{ij})l Z\in \mathrm{Irr}(\mathcal{F}_{kj})}
\sum_{\alpha \in \Lambda}
\frac{\sqrt{d_Z}}{\sqrt{d_Y}}\,
\begin{tikzpicture}[baseline={(current bounding box.center)},scale=0.7]
    \draw[thick]
        (0,0) node[below] {$\scriptstyle Y$}
            -- (0,4) node[above] {$\scriptstyle Z$}
        (2,0) node[below] {$\scriptstyle \overline{Y}$}
            -- (2,4) node[above] {$\scriptstyle \overline{Z}$}
        (0,2) to[out=110,in=270] (-1,4) node[above] {$\scriptstyle W$}
        (2,2) to[out=290,in=90] (3,0) node[below] {$\scriptstyle W$};

    \draw[thick, fill=white] (0,2) circle (0.07);
    \draw[thick, fill=white] (2,2) circle (0.07);

    \node[left]  at (0,2) {$\scriptstyle \alpha$};
    \node[right] at (2,2) {$\scriptstyle \alpha^{\bullet}$};
    \node at (1,2) {$\scriptstyle j$};
    \node at (-1,1) {$\scriptstyle i$};
    \node at (3,3) {$\scriptstyle k$};
    \node at (2.4,0.6) {$\scriptstyle i$};
    \node at (-0.4,3.4) {$\scriptstyle k$};
\end{tikzpicture}
$$

In particular, $L_{j}\otimes \mathbbm{1}_{i}=\mathbbm{1}_{i}\otimes L_{j}=L_{ij}$, then setting $W_{ik}=\bigoplus_{W\in \mathcal{F}_{ik}} K\boxtimes W$ as above, we see that in $\text{Hilb}(\mathcal{F})$, we have the unitary isomorphism

$$\sigma_{L_{j}, W_{ik}}: L_{ij}\otimes W_{ik}\cong W_{ik}\otimes L_{kj},$$

\begin{center}
\begin{minipage}[c]{0.1cm}
    $$\sigma_{L_{j},W_{ik}}:=$$
\end{minipage}
\begin{minipage}[c]{0.3\textwidth}
    $$
\begin{tikzpicture}[scale=0.7]
  
  \draw[black, line width=1.2pt] (0,2) .. controls (0,1.2) and (1,0.8) .. (1,0);

  \draw[red, line width=1.2pt] (1,2) .. controls (1,1.2) and (0,0.8) .. (0,0);

  \node[below] at (0,0) {$L_{ij}$};
  \node[below] at (1,0) {$W_{ik}$};
  \node[above] at (0,2) {$W_{ik}$};
  \node[above] at (1,2) {$L_{kj}$};
\end{tikzpicture}
$$
\end{minipage}
\end{center}

\noindent which is compatible with the algebra structures on $L_{ij}$ and $L_{kj}$ (since they are just restrictions to a direct summand of the central algebra structure on $L_{j}$).

\bigskip

\begin{thm}\label{thm:stableequivsymm}
Let $B_{1}\subseteq A_{1}$ and $B_{2}\subseteq A_{2}$ be anyon chain models with $\mathcal{C}$-symmetry. Then there is a bounded spread isomorphism $\alpha : A_{1}^{\infty} \cong A_{2}^{\infty}$ such that $\alpha(B^{\infty}_{1})=B^{\infty}_{2}$.
\end{thm}

 \begin{proof} Let $\mathcal{C}$ be a unitary fusion category, and suppose $\mathcal{M}_{1}$ and $\mathcal{M}_{2}$ are indecomposable right $\mathcal{C}$-module categories, with duals $\mathcal{D}_{1}$ and $\mathcal{D}_{2}$ respectively.  

Consider the $3\times 3$ indecomposable multi-fusion category 

$$\mathcal{F}:=\text{End}_{\mathcal{C}}(\mathcal{M}_{1}\oplus\mathcal{M}_{2}\oplus \mathcal{C})\cong\begin{bmatrix}
\text{Fun}_{\mathcal{C}}(\mathcal{M}_1, \mathcal{M}_1)  & \text{Fun}_{ \mathcal{C}}(\mathcal{M}_{2}, \mathcal{M}_{1}) & \text{Fun}_{\mathcal{C}}(\mathcal{C}, \mathcal{M}_{1})  \\
\text{Fun}_{ \mathcal{C}}(\mathcal{M}_{1}, \mathcal{M}_{2}) & \text{Fun}_{\mathcal{C}}(\mathcal{M}_1, \mathcal{M}_1) & \text{Fun}_{\mathcal{C}}(\mathcal{C}, \mathcal{M}_{2}) \\
\text{Fun}_{\mathcal{C}}(\mathcal{M}_1, \mathcal{C})  & \text{Fun}_{\mathcal{C}}(\mathcal{M}_2, \mathcal{C}) & \text{Fun}_{\mathcal{C}}(\mathcal{C}, \mathcal{C}) 
\end{bmatrix}$$
$$
\cong \begin{bmatrix}
\mathcal{D}_{1}  & \mathcal{N} & \mathcal{M}_{1} \\
\overline{\mathcal{N}} & \mathcal{D}_{2} & \mathcal{M}_{2}  \\
 \overline{\mathcal{M}}_{1}  &  \overline{\mathcal{M}}_{2} & \mathcal{C}
\end{bmatrix}
$$

\noindent where $\mathcal{N}$ is the invertible $\mathcal{D}_{1}$-$\mathcal{D}_{2}$ bimodule category $\text{Fun}_{ \mathcal{C}}(\mathcal{M}_{2}, \mathcal{M}_{1})\cong \mathcal{M}_{1}\boxtimes_{\mathcal{C}}\overline{\mathcal{M}}_{2}$ \cite{1406.4204}.

Let $A_{1}\in \mathcal{Z}(\mathcal{D}_{1})$ and $A_{2}\in \mathcal{Z}(\mathcal{D}_{2})$ be the Lagrangian algebras associated to the $\mathcal{D}_{i}$-module categories $\mathcal{M}_{i}$. Then by construction, under the identification $\mathcal{D}_{i}\cong \mathcal{F}_{ii}$ for $i=1,2$, we see that $A_{1}\cong L_{13}$ and $A_{2}\cong L_{23}$.

Furthermore, pick strongly tensor generating objects $X_{1}\in \mathcal{D}_{1}$ and $X_{2}\in \mathcal{D}_{2}$, respectively. If we let $\mathcal{E}_{1}=\text{End}(\mathcal{M}_{1})$ and $\mathcal{E}_{2}=\text{End}(\mathcal{M}_{2})$ be the Morita trivial indecomposable multi-fusion categories, and $\mathcal{D}_{i}\rightarrow \mathcal{E}_{i}$ are the forgetful functors from $\text{End}_{\mathcal{C}}(\mathcal{M}_{i})$ to $\text{End}(\mathcal{M}_{i})$, then the corresponding anyon chain models for $\mathcal{C}$ are given by the physical boundary subalgebras

$$A(\mathcal{D}_{i},X_{i})\hookrightarrow A(\mathcal{E}_{i},X_{i}).$$

We recall the \textit{algebra model} for these inclusions, introduced in \cite[Section 4]{jones2025quantumcellularautomatacategorical}. In this picture, we view $ A(\mathcal{E}_{i},X_{i})$ as a DHR bimodule over $A(\mathcal{D}_{i},X_{i})$, which under the equivalence $\text{DHR}(A(\mathcal{D}_{i},X_{i}))\cong \mathcal{Z}(\mathcal{D})$ corresponds to the Lagrangian algebra $A_{i}=L_{i3}$. In particular, we can view the local algebras

$$\begin{tikzpicture}[scale=0.7]

  \node at (-3.1,0) {$A(\mathcal{E}_{i},X_{i})_{I}\cong \text{Hom}_{\mathcal{D}_{i}}(X^{\otimes I}_{i}, X^{\otimes I}_{i}\otimes L_{i3})\ \ \ \leftrightsquigarrow$};

  \begin{scope}[xshift=4cm]
  \draw[thick] (-0.6,0.6) rectangle (0.6,-0.6);
  \node at (0,0) {$\xi$};

  \draw[thick] (-0.45,1.4) -- (-0.45,0.6);
  \draw[thick] (-0.15,1.4) -- (-0.15,0.6);
  \draw[thick] (0.15,1.4) -- (0.15,0.6);
  \draw[thick] (0.45,1.4) -- (0.45,0.6);

  \draw[thick] (-0.45,-0.6) -- (-0.45,-1.4);
  \draw[thick] (-0.15,-0.6) -- (-0.15,-1.4);
  \draw[thick] (0.15,-0.6) -- (0.15,-1.4);
  \draw[thick] (0.45,-0.6) -- (0.45,-1.4);

  \draw[red, thick] (0.6,0.6) -- (1.05,1.4);
  \node[above] at (1.05,1.4) {$L_{i3}$};

  \draw[decorate,decoration={brace,mirror,amplitude=5pt}]
    (-0.55,-1.55) -- (0.55,-1.55)
    node[midway,below=6pt] {$I$};
\end{scope}
\end{tikzpicture}
$$

$$
\begin{tikzpicture}
\begin{scope}[scale=0.8, xshift=0cm]
\node at (-2.1,0) {$\nu\cdot \eta=$};
  
  \draw[thick] (-0.6,1.8) rectangle (0.6,0.6);
  \node at (0,1.2) {$\nu$};

  \draw[thick] (-0.6,0.3) rectangle (0.6,-0.9);
  \node at (0,-0.3) {$\xi$};

  \draw[thick] (-0.45,2.6) -- (-0.45,1.8);
  \draw[thick] (-0.15,2.6) -- (-0.15,1.8);
  \draw[thick] (0.15,2.6) -- (0.15,1.8);
  \draw[thick] (0.45,2.6) -- (0.45,1.8);

  \draw[thick] (-0.45,0.6) -- (-0.45,0.3);
  \draw[thick] (-0.15,0.6) -- (-0.15,0.3);
  \draw[thick] (0.15,0.6) -- (0.15,0.3);
  \draw[thick] (0.45,0.6) -- (0.45,0.3);

  \draw[thick] (-0.45,-0.9) -- (-0.45,-1.7);
  \draw[thick] (-0.15,-0.9) -- (-0.15,-1.7);
  \draw[thick] (0.15,-0.9) -- (0.15,-1.7);
  \draw[thick] (0.45,-0.9) -- (0.45,-1.7);

\draw[red, thick] (0.6,1.8)--(1.8,2.3);
\draw[red, thick] (0.6,0.3)--(1.8,2.3);
\draw[red, thick] (1.8,2.3)--(2.5,3);
\node at (2.5,3.3) {$L_{i3}$};

  \draw[decorate,decoration={brace,mirror,amplitude=5pt}]
    (-0.55,-1.85) -- (0.55,-1.85)
    node[midway,below=6pt] {$I$};
\end{scope}
\end{tikzpicture}
$$

\noindent In this model, the inclusion $A(\mathcal{D}_{i},X_{i})\hookrightarrow A(\mathcal{E}_{i},X_{i})$ is represented graphically by 

$$\begin{tikzpicture}[scale=0.7]

\begin{scope}
  \begin{scope}[xshift=0cm]
  \draw[thick] (-0.6,0.6) rectangle (0.6,-0.6);
  \node at (0,0) {$\xi$};

  \draw[thick] (-0.45,1.4) -- (-0.45,0.6);
  \draw[thick] (-0.15,1.4) -- (-0.15,0.6);
  \draw[thick] (0.15,1.4) -- (0.15,0.6);
  \draw[thick] (0.45,1.4) -- (0.45,0.6);

  \draw[thick] (-0.45,-0.6) -- (-0.45,-1.4);
  \draw[thick] (-0.15,-0.6) -- (-0.15,-1.4);
  \draw[thick] (0.15,-0.6) -- (0.15,-1.4);
  \draw[thick] (0.45,-0.6) -- (0.45,-1.4);

  \draw[decorate,decoration={brace,mirror,amplitude=5pt}]
    (-0.55,-1.55) -- (0.55,-1.55)
    node[midway,below=6pt] {$I$};

    \node  at (2,0) {$\mapsto$};
    
\end{scope}

\end{scope}

  \begin{scope}[xshift=4cm]
  \draw[thick] (-0.6,0.6) rectangle (0.6,-0.6);
  \node at (0,0) {$\xi$};

  \draw[thick] (-0.45,1.4) -- (-0.45,0.6);
  \draw[thick] (-0.15,1.4) -- (-0.15,0.6);
  \draw[thick] (0.15,1.4) -- (0.15,0.6);
  \draw[thick] (0.45,1.4) -- (0.45,0.6);

  \draw[thick] (-0.45,-0.6) -- (-0.45,-1.4);
  \draw[thick] (-0.15,-0.6) -- (-0.15,-1.4);
  \draw[thick] (0.15,-0.6) -- (0.15,-1.4);
  \draw[thick] (0.45,-0.6) -- (0.45,-1.4);

  \draw[red, thick] (0.8,1) -- (1.05,1.4);
  \node[above] at (1.05,1.4) {$L_{i3}$};
  \node [circle, fill=red, scale=0.3] at (0.8,1) {};

  \draw[decorate,decoration={brace,mirror,amplitude=5pt}]
    (-0.55,-1.55) -- (0.55,-1.55)
    node[midway,below=6pt] {$I$};
\end{scope}
\end{tikzpicture}
$$

Now, recall the bounded spread isomorphism $\alpha: A(\mathcal{D}_{1}, X_{1})^{\infty}\cong A(\mathcal{D}_{2}, X_{2})^{\infty}$ from Theorem \ref{thm:MoritaEquivimpliesStableEquiv}. Then using $\sigma_{L_{13}, W_{12}}: L_{13}\otimes W_{12}\cong W_{12}\otimes L_{23}$, followed by the half-braiding $L_{23}\otimes W_{22}\cong W_{22}\otimes L_{23}$, we obtain a canonical extension of $\alpha$ to a bounded spread isomorphism $\widetilde{\alpha}: A(\mathcal{E}_{1}, X_{1})^{\infty}\cong A(\mathcal{E}_{2}, X_{2})^{\infty}$ by

$$
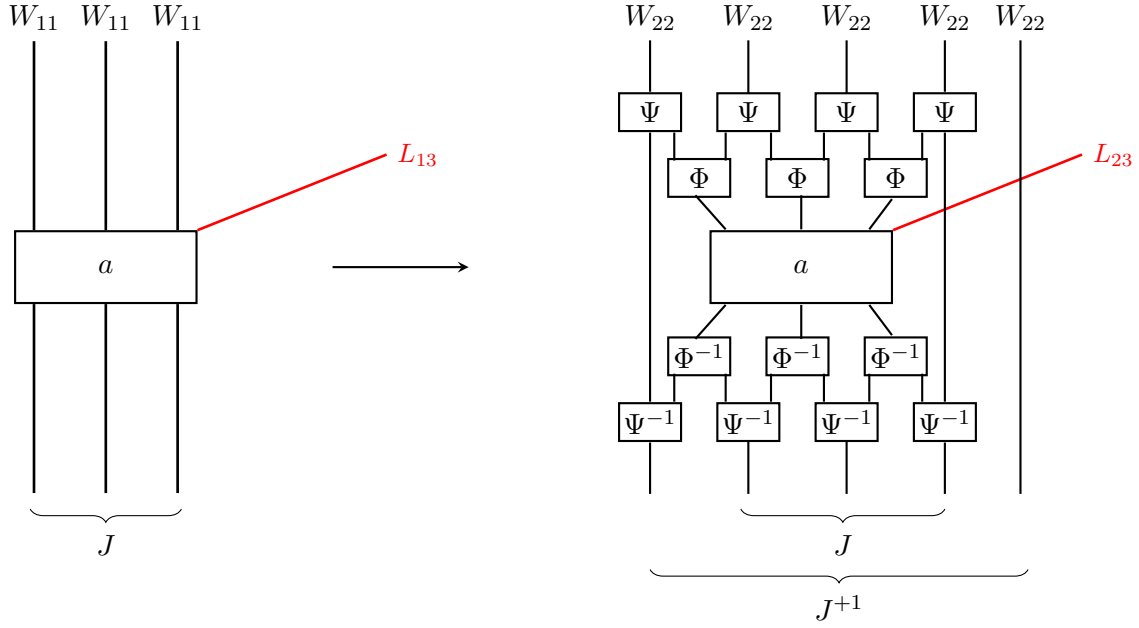
\begin{tikzpicture}[x=1cm,y=1cm,>=stealth]
\tikzset{
  wire/.style={line width=1pt},
  bigbox/.style={draw, thick, minimum width=2.4cm, minimum height=0.95cm},
  smallbox/.style={draw, thick, minimum width=0.82cm, minimum height=0.50cm, inner sep=1pt},
  lab/.style={font=\small}
}

\node[bigbox] (fL) at (0,0) {$a$};

\foreach \x in {-0.95,0,0.95}{
  \draw[wire] ($(fL.north)+(\x,0)$) -- ++(0,2.5);
  \draw[wire] ($(fL.south)+(\x,0)$) -- ++(0,-2.5);
}
\draw[wire,red] (fL.north east) -- ++(2.5,1.0)
  node[right,lab] {$L_{13}$};

    \node at (-0.95,3.3) {$W_{11}$};
    \node at (0,3.3) {$W_{11}$};
    \node at (0.95,3.3) {$W_{11}$};
      \draw[decorate,decoration={brace,mirror,amplitude=5pt}]
    (-1,-3.2) -- (1,-3.2)
    node[midway,below=6pt] {$J$};

\draw[->, thick] (3.0,0) -- (4.8,0);

\node[bigbox] (fR) at (9.2,0) {$a$};

\draw[wire,red] (fR.north east) -- ++(2.5,1.0)
  node[right,lab] {$L_{23}$};

\def\xA{7.2}
\def\xB{8.5}
\def\xC{9.8}
\def\xD{11.1}

\def\xAB{7.85}
\def\xBC{9.15}
\def\xCD{10.45}

\def\yPsi{2.05}
\def\yPhi{1.18}
\def\yIphi{-1.18}
\def\yIpsi{-2.05}

\draw[thick] (7.2,2.32) -- (7.2,3);
\draw[thick] (7.2,1.78) -- (7.2,-1.78);
\draw[thick] (7.2,-2.32) -- (7.2,-3);

\draw[thick] (11.1,2.32) -- (11.1,3);
\draw[thick] (11.1,1.78) -- (11.1,-1.78);
\draw[thick] (11.1,-2.32) -- (11.1,-3);

\draw[thick] (7.52,1.78) -- (7.52,1.41);
\draw[thick] (7.52,-1.78) -- (7.52,-1.41);

\draw[thick] (8.8,1.78) -- (8.8,1.41);
\draw[thick] (9.2,0.5) -- (9.2,0.95);
\draw[thick] (9.2,-0.5) -- (9.2,-0.95);
\draw[thick] (8.8,-1.78) -- (8.8,-1.41);

\draw[thick] (8.2,1.78) -- (8.2,1.41);
\draw[thick] (7.8,0.95) -- (8.2,0.5);
\draw[thick] (8.2,-0.5) -- (7.8,-0.95);
\draw[thick] (8.2,-1.78) -- (8.2,-1.41);

\draw[thick] (8.5,2.3) -- (8.5,3);
\draw[thick] (8.5,-2.3) -- (8.5,-3);

\draw[thick] (9.8,2.3) -- (9.8,3);
\draw[thick] (9.5,1.78) -- (9.5,1.41);
\draw[thick] (9.5,-1.78) -- (9.5,-1.41);
\draw[thick] (9.8,-2.3) -- (9.8,-3);

\draw[thick] (10.1,1.78) -- (10.1,1.41);
\draw[thick] (10.1,-1.78) -- (10.1,-1.41);

\draw[thick] (10.8,1.78) -- (10.8,1.41);
\draw[thick] (10.8,-1.78) -- (10.8,-1.41);

\draw[thick] (10.1,0.5) -- (10.4,0.9);
\draw[thick] (10.1,-0.5) -- (10.4,-0.9);

\draw[thick] (12.1,-3) -- (12.1,3);

\node[smallbox] (psi1) at (\xA,\yPsi) {$\Psi$};
\node[smallbox] (psi2) at (\xB,\yPsi) {$\Psi$};
\node[smallbox] (psi3) at (\xC,\yPsi) {$\Psi$};
\node[smallbox] (psi4) at (\xD,\yPsi) {$\Psi$};

\node[smallbox] (phi1) at (\xAB,\yPhi) {$\Phi$};
\node[smallbox] (phi2) at (\xBC,\yPhi) {$\Phi$};
\node[smallbox] (phi3) at (\xCD,\yPhi) {$\Phi$};

\node[smallbox] (iphi1) at (\xAB,\yIphi) {$\Phi^{-1}$};
\node[smallbox] (iphi2) at (\xBC,\yIphi) {$\Phi^{-1}$};
\node[smallbox] (iphi3) at (\xCD,\yIphi) {$\Phi^{-1}$};

\node[smallbox] (ipsi1) at (\xA,\yIpsi) {$\Psi^{-1}$};
\node[smallbox] (ipsi2) at (\xB,\yIpsi) {$\Psi^{-1}$};
\node[smallbox] (ipsi3) at (\xC,\yIpsi) {$\Psi^{-1}$};
\node[smallbox] (ipsi4) at (\xD,\yIpsi) {$\Psi^{-1}$};

 \node at (7.2,3.3) {$W_{22}$};
    \node at (8.5,3.3) {$W_{22}$};
    \node at (9.8,3.3) {$W_{22}$};
    \node at (11.1,3.3) {$W_{22}$};
    \node at (12.1,3.3) {$W_{22}$};
    \draw[decorate,decoration={brace,mirror,amplitude=5pt}]
    (8.4,-3.2) -- (11.1,-3.2)
    node[midway,below=6pt] {$J$};
      \draw[decorate,decoration={brace,mirror,amplitude=5pt}]
    (7.2,-4) -- (12.2,-4)
    node[midway,below=6pt] {$J^{+1}$};
\end{tikzpicture}.
$$

\end{proof}

If $B\subseteq A$ is a symmetry inclusion, a (locally normal) state $\phi$ on $A$ is called \textit{symmetric} if $\phi\circ E=\phi$. Consider the Hilbert space representation of $B$ on the GNS representation $L^{2}(B,\phi)$. This generates a module category of $\text{DHR}_{0}(B)$, $\mathcal{M}_{\phi}:=\text{DHR}_{0}(B)\boxtimes_{B} L^{2}(B,\phi)$, and taking internal End of the object $L^{2}(B,\phi)\in \mathcal{M}_{\phi}$, we obtain a W*-algebra object $L_{\phi}\in \text{DHR}_{0}(B)$. We call this the \textit{symmetry order algebra} (cf. \cite{jones2025localtopologicalorderboundary, evans2026operatoralgebraicapproachfusion}). If $L_{\phi}$ is Lagrangian, we say $L_{\phi}$ is \textit{topological}, since it corresponds to a topological boundary condition on the SymTFT corresponding to $\text{DHR}_{0}(B)$.

If $\alpha$ is an equivalence $B_{1}\subseteq A_{1}\cong B_{2}\subseteq A_{2}$, then precomposition with $\alpha^{-1}$  establishes a continuous bijection between locally normal symmetric states on $A_{1}$ and locally normal symmetric states on $A_{2}$. Furthermore $\alpha|_{B_{1}}$ yields an equivalence $\widetilde{\alpha}: \text{DHR}_{0}(B_1)\cong \text{DHR}_{0}(B_2)$, and if $\phi_{2}\circ \alpha $ is unitarily equivalent to $\phi_{1}$, then under this equivalence $\widetilde{\alpha}(L_{\phi_1})\cong L_{\phi_2}$.

Now suppose $B\subseteq A$ is an anyon chain, and let $\phi$ be any symmetric state on $B\subseteq A$. Then the stabilization $\phi^{\infty}$ is the locally normal symmetric state on $B^{\infty}\subseteq A^{\infty}$ defined by $$\phi^{\infty}:=\phi\otimes \left(|\xi\rangle^{\overline{\otimes} \mathbbm{Z}} \right)$$

\noindent where $|\xi\rangle\in \ell^{2}(\mathbbm{N})$ is any vector state on $\mathcal{B}(\ell^{2}(\mathbbm{N}))$, and $|\xi\rangle^{\overline{\otimes} \mathbbm{Z}}$ is the corresponding pure (locally normal) product state on $\mathcal{L}^{\infty}$. Then under the equivalence $\text{DHR}(B)\cong \text{DHR}_{0}(B^{\infty})$, we see that $L_{\phi}\cong L_{\phi^{\infty}}$. Taken together we have the following:

\begin{thm}\label{thm:symmetryorderinvariance}
If $B\subseteq A$ is an anyon chain, and $\phi$ is a symmetric state, then the symmetry order algebra $L_{\phi}$ is an invariant under stable equivalence.
\end{thm}

\bigskip

\subsection{Classifying physical boundary algebras}

Quasi-local algebras like fusion spin chains, with nontrivial DHR bimodule categories, arise not only in the context of categorical symmetry, but also as operators localized near the physical boundary of a topologically ordered system \cite{jones2025localtopologicalorderboundary, jones2025holographybulkboundarylocaltopological, ChuahHungarKawagoePenneysTombaWallickWei2024BoundaryAlgebras}. This is the essence of topological holography.

Fusion spin chains arise as the boundary algebras of commuting projector models, e.g. Levin-Wen string-net models. However, the specific fusion spin chain is very sensitive to the details of the microscopic model. The following result shows that after stabilization, the bounded spread isomorphism class of boundary algebras depends only on the bulk topological order.

\begin{cor}
    Let $A:=A(\mathcal{C},X)$ and $B:=A(\mathcal{D},Y)$ be fusion spin chains, with $\mathcal{C}$ and $\mathcal{D}$ indecomposable unitary multi-fusion categories, and $X$ and $Y$ strong tensor generators. Then $A^{\infty}$ is bounded spread isomorphic to $B^{\infty}$ if and only if $\mathcal{C}$ is Morita equivalent to $\mathcal{D}$ as multi-fusion categories.
\end{cor}

\bibliographystyle{unsrt}
\bibliography{bibliography}

\end{document}